\documentclass[11pt]{article}
\usepackage{hyperref}
\usepackage{times}  % Fonts
\usepackage{mathpazo}
\usepackage{amssymb,amsmath,amsthm}
\usepackage{epsfig}
\usepackage{tcolorbox}
\usepackage{lipsum} % for dummy text
\usepackage{enumitem}

\tcbset{ % set some default parameters for the boxes
  halign=justify, % justify the text inside the box
  center, % center the box on the page
  colback=gray!10, % set the background color of the box to a very light gray
  colframe=black, % set the frame color of the box
  boxrule=0.5pt % set the width of the boundary of the box
}

\newcommand{\cd}{\mathop{\mathrm{cd}}}
\newcommand{\alt}{\mathop{\mathrm{alt}}}

\newtheorem{theorem}{Theorem}[section]

\newtheorem{definition}[theorem]{Definition}

\newtheorem{lemma}[theorem]{Lemma}

\newtheorem{corollary}[theorem]{Corollary}

\def\lsoft{{l\kern-0.035cm\char39\kern-0.03truecm}}

\newcommand\kriz{K\v{r}\'{\i}\v{z}}
\newcommand\barany{B\'ar\'any}

\newcommand{\qedsymb}{\hfill{\rule{2mm}{2mm}}}
\renewenvironment{proof}[1][]{\begin{trivlist}
\item[\hspace{\labelsep}{\bf\noindent Proof#1:\/}] }{\qedsymb\end{trivlist}}

\def\calA{{\cal A}}
\def\calB{{\cal B}}
\def\calG{{\cal G}}
\def\calF{{\cal F}}

\def\Z{{\mathbb{Z}}}
\def\R{\mathbb{R}}

\def\N{\mathbb{N}}

\newcommand{\SchrijverP}{\textsc{Schrijver}}
\newcommand{\KneserP}{\textsc{Kneser}}
\newcommand{\TwoStab}{\mathrm{stab}}
\newcommand{\AlmostTwoStab}{\widetilde{\mathrm{stab}}}
\newcommand{\KneserPstab}{\textsc{Kneser}_{\TwoStab}}
\newcommand{\KneserPAstab}{\textsc{Kneser}_{\AlmostTwoStab}}

\newcommand{\ZpTucker}{\Z_p\textsc{-Tucker}}

\newcommand{\PPA}{\mathsf{PPA}}
\newcommand{\PPAD}{\mathsf{PPAD}}

\newcommand{\PPAp}{\mathsf{PPA}\textsc{-}p}
\newcommand{\PPAr}{\mathsf{PPA}\textsc{-}r}
\newcommand{\PPAthree}{\mathsf{PPA}\textsc{-}3}
\newcommand{\PPAtwo}{\mathsf{PPA}\textsc{-}2}
\newcommand{\TFNP}{\mathsf{TFNP}}

\newcommand{\ConHalv}{\textsc{Con-Halving}}

\newcommand{\threeConDiv}{\textsc{Con-$3$-Division}}
\newcommand{\pConDiv}{\textsc{Con-$p$-Division}}
\newcommand{\rConDiv}{\textsc{Con-$r$-Division}}

\newcommand\ConDiv[2]{\textsc{Con-$#2$-Division}[<#1]}
\newcommand\ConHalvi[1]{\textsc{Con-Halving}[<#1]}

\newcommand{\NP}{\mathsf{NP}}

\newcommand{\eps}{\epsilon}
\renewcommand{\epsilon}{\varepsilon}

\begin{document}

\title{{\bf The Chromatic Number of Kneser Hypergraphs via Consensus Division}}
%\footnote{A preliminary version appeared in Proc. of the 15th Innovations in Theoretical Computer Science conference (ITCS), pages 60:1--17, 2024.}}

\author{
Ishay Haviv\thanks{School of Computer Science, The Academic College of Tel Aviv-Yaffo, Tel Aviv 61083, Israel. Research supported in part by the Israel Science Foundation (grant No.~1218/20).}
}

\date{}

\maketitle

\begin{abstract}
We show that the Consensus Division theorem implies lower bounds on the chromatic number of Kneser hypergraphs, offering a novel proof for a result of Alon, Frankl, and Lov{\'{a}}sz (Trans.~Amer.~Math.~Soc.,~1986) and for its generalization by \kriz~(Trans.~Amer.~Math.~Soc.,~1992).
Our approach is applied to study the computational complexity of the total search problem $\KneserP^p$, which given a succinct representation of a coloring of a $p$-uniform Kneser hypergraph with fewer colors than its chromatic number, asks to find a monochromatic hyperedge.
We prove that for every prime $p$, the $\KneserP^p$ problem with an extended access to the input coloring is efficiently reducible to a quite weak approximation of the Consensus Division problem with $p$ shares. In particular, for $p=2$, the problem is efficiently reducible to {\em any non-trivial} approximation of the Consensus Halving problem on normalized monotone functions.
We further show that for every prime $p$, the $\KneserP^p$ problem lies in the complexity class $\PPAp$. As an application, we establish limitations on the complexity of the $\KneserP^p$ problem, restricted to colorings with a bounded number of colors.
\end{abstract}

%\newpage

\pagenumbering{arabic}

\section{Introduction}

This paper is concerned with two classic problems: the graph-theoretic problem of determining the {\em chromatic number of Kneser hypergraphs} and the {\em consensus division} problem from the area of fair division, which lies at the intersection of economics, mathematics, and computer science.
We present a novel direct connection between the two problems, offering a new proof for a result of Alon, Frankl, and Lov{\'{a}}sz~\cite{AlonFL86} on the chromatic number of Kneser hypergraphs as well as for its generalization by \kriz~\cite{Kriz92}. We use this connection to study the computational complexity of the total search problems associated with Kneser hypergraphs and with approximate consensus division. In what follows, we provide some background on the two problems and on their computational aspects, and then describe our contribution.

\paragraph{Kneser hypergraphs.}
For an integer $r \geq 2$ and a set family $\calF$, the $r$-uniform Kneser hypergraph $K^r(\calF)$ is the hypergraph on the vertex set $\calF$, whose hyperedges are all the $r$-subsets of $\calF$ whose members are pairwise disjoint. For integers $n$ and $k$ with $n \geq r \cdot k$, let $K^r(n,k)$ denote the hypergraph $K^r(\calF)$, where $\calF = \binom{[n]}{k}$ is the family of all $k$-subsets of $[n] = \{1,2,\ldots,n\}$. When $r=2$, the superscript $r$ may be omitted. The chromatic number of a hypergraph $H$, denoted by $\chi(H)$, is the minimum number of colors that allow a proper coloring of its vertices, that is, a coloring with no monochromatic hyperedge.

The study of the chromatic number of Kneser graphs was initiated in 1955 by Kneser~\cite{Kneser55}, who observed that the graph $K(n,k)$ admits a proper coloring with $n-2k+2$ colors and conjectured that fewer colors do not suffice, that is, $\chi(K(n,k)) = n-2k+2$. Lov{\'{a}}sz~\cite{LovaszKneser} confirmed Kneser's conjecture in 1978, and his result was extended in multiple ways over the years. One extension, due to Schrijver~\cite{SchrijverKneser78}, showed that the subgraph $S(n,k)$ of $K(n,k)$, induced by the $k$-subsets of $[n]$ that are stable (i.e., include no two consecutive elements modulo $n$) has the same chromatic number.
Another extension, due to Alon, Frankl, and Lov{\'{a}}sz~\cite{AlonFL86}, confirmed a conjecture of Erd{\H{o}}s~\cite{Erdos76} by showing that $\chi(K^r(n,k)) = \lceil \frac{n-r(k-1)}{r-1} \rceil$ for all integers $r \geq 2$. Their lower bound on $\chi(K^r(n,k))$ was further generalized by \kriz~\cite{Kriz92}, as stated below, using a quantity of set families, called the $r$-colorability defect and denoted by ${\cd}_r$ (see Definition~\ref{def:color_defect}; see also~\cite{Dolnikov82}).
\begin{theorem}[\cite{Kriz92}]\label{thm:IntroKriz}
For every integer $r \geq 2$ and for every family $\calF$ of non-empty sets,
\[ \chi(K^r(\calF)) \geq \Big \lceil \frac{\cd_r(\calF)}{r-1} \Big \rceil.\]
\end{theorem}
\noindent
While Theorem~\ref{thm:IntroKriz} implies the tight lower bound of Alon et al.~\cite{AlonFL86} on the chromatic number of $K^r(n,k)$ (see Lemma~\ref{lemma:cd_Kneser}), it does not cover the aforementioned result of Schrijver~\cite{SchrijverKneser78}.
In an attempt to simultaneously generalize both the results, additional extensions were established by several authors, e.g., Meunier~\cite{Meunier11}, Alishahi and Hajiabolhassan~\cite{AlishahiH15}, Aslam et al.~\cite{AslamCCFS19}, and Frick~\cite{Frick20}.

It is interesting to mention that although the statements of the above results are all purely combinatorial, their proofs rely on topological tools. Lov{\'{a}}sz's lower bound~\cite{LovaszKneser} on the chromatic number of $K(n,k)$ was based on the Borsuk--Ulam theorem~\cite{Borsuk33}, a fundamental result in algebraic topology, and his approach pioneered the area of topological combinatorics. The extension to hypergraphs by Alon et al.~\cite{AlonFL86} was based on a theorem of \barany, Shlosman, and Sz{\H{u}}cs~\cite{BSS81} that generalizes the Borsuk--Ulam theorem. It was shown by Matou{\v{s}}ek~\cite{Matousek04} that the chromatic number of $K(n,k)$ can also be determined as an application of Tucker's lemma, a combinatorial analogue of the Borsuk--Ulam theorem. His machinery was further developed by Ziegler~\cite{Ziegler02} and later by Meunier~\cite{Meunier11} to provide alternative proofs for the results of~\cite{AlonFL86,Kriz92,SchrijverKneser78}.

\paragraph{Consensus division.}

Another area that extensively applies topological tools is fair division, where the goal is to find fair allocations of resources among several parties. In the consensus division scenario, given $m$ continuous valuation functions $v_1, \ldots, v_m$ defined on subsets of the unit interval $[0,1]$, we aim to divide the interval into $r$ (not necessarily connected) pieces $A_1, \ldots, A_r$ using as few cuts as possible, such that each function assigns the same value to the $r$ pieces, namely, $v_i(A_{t}) = v_i(A_{t'})$ for all $i \in [m]$ and $t,t' \in [r]$. For the case $r=2$, referred to as consensus halving, the Hobby--Rice theorem~\cite{HobbyRice65} asserts that for additive valuation functions, there always exists such a division with at most $m$ cuts. For a general $r \geq 2$, Alon~\cite{Alon87Necklace} used the generalization of~\cite{BSS81} of the Borsuk--Ulam theorem to show that if the valuation functions are probability measures on $[0,1]$, then there exists a consensus division with at most $(r-1) \cdot m$ cuts. Note that this result played a central role in his proof of the Splitting Necklaces theorem~\cite{Alon87Necklace}. In fact, it turns out that if the number of desired pieces is a prime $p$, then the existence of a consensus division with at most $(p-1) \cdot m$ cuts is guaranteed for any given continuous valuation functions, which unlike probability measures, are not required to be additive nor non-negative. This extension was provided for $p=2$ by Simmons and Su~\cite{SimmonsS03} and for any prime $p$ by Filos-Ratsikas, Hollender, Sotiraki, and Zampetakis~\cite{FHSZ21} (see Theorem~\ref{thm:ConDiv}).

\paragraph{Computational aspects.}
The complexity class $\TFNP$, introduced in~\cite{MegiddoP91}, consists of the total search problems in $\NP$, i.e., the search problems for which every instance is guaranteed to have a solution, where the solutions can be verified in polynomial time. In 1994, Papadimitriou~\cite{Papa94} introduced several sub-classes of $\TFNP$ that express the mathematical arguments which yield the existence of the solutions for their problems.
In particular, he defined for every prime $p$, the complexity class $\PPAp$, associated with the principle that every bipartite graph that has a vertex whose degree is not a multiple of $p$, has another such vertex. Consequently, $\PPAp$ is the class of problems in $\TFNP$ that are efficiently reducible to the corresponding canonical problem: Given a Boolean circuit that represents a bipartite graph and given a vertex whose degree is not a multiple of $p$, find another such vertex.
This family of classes was further studied and developed by Hollender~\cite{Hollender21} and by G{\"{o}}{\"{o}}s, Kamath, Sotiraki, and Zampetakis~\cite{GoosKSZ20}, who independently introduced the classes $\PPAr$ for composites $r$.
Note that the class $\PPAtwo$ is denoted by $\PPA$, and that its analogue for directed graphs, called  $\PPAD$, is known to be contained in $\PPAr$ for all integers $r \geq 2$~\cite{Hollender21,GoosKSZ20}.

For any integer $r \geq 2$, let $\rConDiv$ denote the computational problem, which given an access to $m$ continuous valuation functions $v_1, \ldots, v_m$ over $[0,1]$, asks to find a consensus division into $r$ pieces $A_1, \ldots, A_r$ with at most $(r-1) \cdot m$ cuts. In its approximate version with precision parameter $\eps$, the pieces should satisfy $|v_i(A_t) - v_i(A_{t'})| \leq \eps$ for all $i \in [m]$ and $t,t' \in [r]$.
The problem may be considered for various families of input functions, ranging from piecewise-constant functions and probability measures to monotone functions and general functions. While piecewise-constant functions can be explicitly represented by the endpoints and values of the intervals on which they are nonzero, in the more general settings the functions are given by some succinct representation, e.g., arithmetic circuits or efficient Turing machines that compute them.
It follows from~\cite{FHSZ21} that for all primes $p$, every instance of the $\pConDiv$ problem has a solution even for $\eps=0$, hence the problem is total.

A considerable attention has been devoted in recent years to the $\rConDiv$ problem with $r=2$, referred to as $\ConHalv$.
Filos-Ratsikas and Goldberg~\cite{FG19} proved that the problem is $\PPA$-complete, even for piecewise-constant functions and inverse-polynomial precision parameter $\eps$, and derived the $\PPA$-completeness of the Splitting Necklaces with two thieves and Discrete Sandwich problems. Their result was strengthened and extended in several ways. Deligkas, Filos-Ratsikas, and Hollender~\cite{DeligkasFH22} proved that the problem remains $\PPA$-hard when the number of valuation functions is a fixed constant: for two or more general functions, and for three or more monotone functions.
Filos-Ratsikas, Hollender, Sotiraki, and Zampetakis~\cite{FHSZ23} proved that the problem remains $\PPA$-hard when the valuation functions are piecewise-uniform with only two blocks. More recently, Deligkas, Fearnley, Hollender, and Melissourgos~\cite{DeligkasFHM22} proved that the problem is $\PPA$-hard for any precision parameter $\eps < 1/5$, even when the valuation functions are piecewise-uniform with three blocks.
An additional version of the $\ConHalv$ problem, where the goal is not to partition an interval but an unordered collection of items, was studied by Goldberg, Hollender, Igarashi, Manurangsi, and Suksompong~\cite{GoldbergHIMS20}.

The complexity of the $\rConDiv$ problem for a general $r$ is much less understood. It was proved by Filos-Ratsikas, Hollender, Sotiraki, and Zampetakis~\cite{FHSZ21} that for every prime $p$, $\pConDiv$ lies in the complexity class $\PPAp$, and the question of whether it is $\PPAp$-hard was left open. Yet, for $p=3$, it was shown in~\cite{FHSZ23} that $\threeConDiv$ is $\PPAD$-hard for an inverse-exponential precision parameter $\eps$.
On the algorithmic side, Alon and Graur~\cite{AlonG21} proved that given an access to $m$ probability measures over $[0,1]$, it is possible to find in polynomial time a partition of the unit interval into $r$ pieces with at most $(r-1) \cdot m$ cuts, such that every piece receives at least $\frac{1}{m \cdot r}$ of each measure.
They also considered the case where a larger number of cuts is allowed, and showed that for a given precision parameter $\eps$, it is possible to efficiently find a solution with $O(\eps^{-2} \cdot m \log m)$ cuts.
Goldberg and Li~\cite{GoldbergL23} proved that the $\rConDiv$ problem on $m$ probability measures, where the number of allowed cuts is $2(r-1)(p-1) \cdot \frac{\lceil p/2 \rceil}{\lfloor p/2 \rfloor} \cdot m$ for a prime $p$, lies in the complexity class $\PPAp$. For example, the $\ConHalv$ problem on $m$ probability measures with $8m$ allowed cuts lies in $\PPAthree$.

Before turning to our results, let us introduce, for any integer $r \geq 2$, the computational search problem $\KneserP^r$ associated with the family of $r$-uniform Kneser hypergraphs $K^r(n,k)$.
The input consists of integers $n$ and $k$ with $n \geq r \cdot k$ along with a Boolean circuit that represents a coloring of the vertices of the hypergraph $K^r(n,k)$ with $\lfloor \frac{n-r(k-1)-1}{r-1} \rfloor$ colors, which is smaller by one than its chromatic number~\cite{AlonFL86}. The goal is to find a monochromatic hyperedge, that is, $r$ pairwise disjoint vertices that are assigned the same color by the input coloring. Omitting the superscript $r$ when $r=2$, the $\KneserP$ problem is known to lie in $\PPA$, and it is an open question whether it is $\PPA$-hard, as suggested by Deng, Feng, and Kulkarni~\cite{DengFK17}. More generally, it was asked in~\cite{FHSZ21} whether for every prime $p$, the $\KneserP^p$ problem lies in $\PPAp$ and if it is $\PPAp$-hard. While no hardness result is known for these problems, it was shown in~\cite{Haviv22-FISC} that the $\SchrijverP$ problem, which asks to find a monochromatic edge in a graph $S(n,k)$ given a coloring of its vertices with $n-2k+1$ colors, is $\PPA$-complete. It was recently shown in~\cite{Haviv23} that the problem of finding a monochromatic edge in a graph $S(n,k)$, given a coloring of its vertices with only $\lfloor n/2 \rfloor -2k+1$ colors, is efficiently reducible to the $\KneserP$ problem.

\subsection{Our Contribution}

This paper presents a novel direct connection between the chromatic number of Kneser hypergraphs and the consensus division problem.
As our first contribution, we offer a new proof of \kriz's lower bound on the chromatic number of Kneser hypergraphs~\cite{Kriz92}, stated earlier as Theorem~\ref{thm:IntroKriz}. The proof relies on the Consensus Division theorem of Filos-Ratsikas et al.~\cite{FHSZ21}. Our technique borrows and extends ideas that were applied in~\cite{GoldbergHIMS20} and in~\cite{Haviv22a}.

We then adopt a computational perspective and explore our approach to Theorem~\ref{thm:IntroKriz} as a reduction from the $\KneserP^p$ problem to the $\pConDiv$ problem for any prime $p$.
Our main contribution is an efficient reduction from the $\KneserP^p$ problem with an extended access to the input coloring to a quite weak approximation of the $\pConDiv$ problem.
Before the precise statements, let us introduce the following variants of the studied computational problems. Their formal definitions are given in Sections~\ref{subsec:Kneser} and~\ref{subsec:CD}.

\begin{itemize}
  \item The $\KneserP^p$ problem with {\em subset queries}: As in the standard $\KneserP^p$ problem, the input is a coloring of the vertices of a hypergraph $K^p(n,k)$ with fewer colors than its chromatic number, and the goal is to find a monochromatic hyperedge. Here, however, the access to the coloring allows, in addition to queries for the colors of the vertices, another type of queries called subset queries. Such a query involves a subset $D$ of $[n]$ and a color $i$, and the answer on the pair $(D,i)$ determines whether $D$ contains a vertex colored $i$. The notion of subset queries was proposed in~\cite{Haviv22a}.
  \item The $\ConDiv{\eps}{p}$ problem: As in the standard $\pConDiv$ problem, the input consists of $m$ continuous valuation functions $v_1, \ldots, v_m$ over $[0,1]$, and a solution is a partition of the unit interval into $p$ pieces $A_1, \ldots, A_p$ with at most $(p-1) \cdot m$ cuts. Here, however, the solution is required to satisfy
        \begin{eqnarray}\label{eq:strict}
            |v_i(A_t) - v_i(A_{t'})| < \eps
        \end{eqnarray}
      for all $i \in [m]$ and $t,t' \in [p]$. Namely, the difference from the standard $\pConDiv$ problem with precision parameter $\eps$ is that the inequality in~\eqref{eq:strict} is strict. In particular, for $\eps=1$, assuming that the valuation functions are normalized (i.e., return values in $[0,1]$), the solution is just required to be {\em non-trivial}. This means that the solution is just not allowed to include two pieces $A_t$ and $A_{t'}$ such that $v_i(A_t)=1$ and $v_i(A_{t'})=0$ for some $i \in [m]$, but the value of $|v_i(A_t)-v_i(A_{t'})|$ may approach $1$ as $m$ grows. When $p=2$, the problem is denoted by $\ConHalvi{\eps}$.
\end{itemize}

We prove the following theorem, which concerns the case $p=2$.
\begin{theorem}\label{thm:IntroKneserToCH}
There exists a polynomial-time reduction from the $\KneserP$ problem with subset queries to the $\ConHalvi{1}$ problem on normalized monotone functions.
\end{theorem}

As alluded to before, the complexity of the $\KneserP$ problem (with or without subset queries) is not well understood.
Theorem~\ref{thm:IntroKneserToCH} implies that any hardness result for the $\KneserP$ problem with subset queries would imply a very strong hardness result for the $\ConHalv$ problem on normalized monotone functions, ruling out the possibility to obtain an efficient algorithm for any non-trivial approximation of the problem. On the other hand, an efficient algorithm for $\ConHalv$ that guarantees some non-trivial approximation on normalized monotone functions would imply an efficient algorithm for the $\KneserP$ problem with subset queries.
We find these consequences of the reduction quite surprising and unusual, especially because of the discrete nature of the $\KneserP$ problem. For comparison, the efficient reduction from the (discrete) Splitting Necklaces problem with two thieves to the $\ConHalv$ problem with precision parameter $\eps$, which is given in~\cite{FG19} and builds on an argument of~\cite{Alon87Necklace}, requires $\eps$ to be inverse-polynomial in the number of valuation functions (note that those functions are additive, though).
Let us stress that Theorem~\ref{thm:IntroKneserToCH} addresses the $\ConHalvi{1}$ problem when restricted to normalized monotone valuation functions but not to probability measures. Recall that for this stronger restriction, an algorithm of~\cite{AlonG21} does provide a non-trivial solution in polynomial time.
We finally note that the proof of Theorem~\ref{thm:IntroKneserToCH} essentially supplies a reduction to the version of $\ConHalv$, studied in~\cite{GoldbergHIMS20}, of finding a consensus halving of an unordered collection of items rather than of the unit interval. This makes the result stronger, as this version is efficiently reducible to the standard one.

Our next result relates the $\KneserP^p$ problem with subset queries to the $\pConDiv$ problem for every prime $p \geq 3$. Here, the precision parameter of the latter is $\tfrac{1}{2}$.

\begin{theorem}\label{thm:IntroKneserToCD}
For every prime $p \geq 3$, there exists a polynomial-time reduction from the $\KneserP^p$ problem with subset queries to the $\ConDiv{\tfrac{1}{2}}{p}$ problem on normalized monotone functions.
\end{theorem}

It is noteworthy that Theorems~\ref{thm:IntroKneserToCH} and~\ref{thm:IntroKneserToCD} are proved in a more general form.
Namely, we reduce from a generalized variant of the $\KneserP^p$ problem, where the input is a coloring of a hypergraph $K^p(\calF)$ for some set family $\calF$ taken from a prescribed sequence of set families, which is assumed to be efficiently computable (see Definitions~\ref{def:KneserProblemSubset} and~\ref{def:computable_seq}). The number of colors used by the input coloring may be any number smaller than \kriz's lower bound on the chromatic number of $K^p(\calF)$, as given by Theorem~\ref{thm:IntroKriz}.
For the precise statements, see Theorems~\ref{thm:Kneser->CH1} and~\ref{thm:Kneser->CD}.

We proceed with additional results on the computational complexity of the $\KneserP^p$ problem (in its standard version, without subset queries).
The following theorem settles a question of~\cite{FHSZ21}.
\begin{theorem}\label{thm:IntroPPAp}
For every prime $p$, the $\KneserP^p$ problem lies in $\PPAp$.
\end{theorem}
In fact, we provide two results that strengthen Theorem~\ref{thm:IntroPPAp} in two incomparable forms.
Firstly, as before, we present a generalized result handling a variant of the $\KneserP^p$ problem of finding a monochromatic hyperedge in a hypergraph $K^p(\calF)$, given a coloring that uses fewer colors than \kriz's lower bound on its chromatic number. Again, $\calF$ is taken from a prescribed sequence of set families satisfying certain computational assumptions (see Definition~\ref{def:computable_seq_order} and Corollary~\ref{cor:Kneser^p_cd_inPPAp}).
Secondly, we show that the membership in $\PPAp$ holds for the $\KneserPAstab^p$ problem. This problem asks to find a monochromatic hyperedge in the sub-hypergraph of $K^p(n,k)$ induced by the $k$-subsets of $[n]$ that are {\em almost stable} (i.e., include no two consecutive elements, but can include both $1$ and $n$). The input coloring may use here any number of colors smaller than the chromatic number of this hypergraph, which, as proved in~\cite{Meunier11}, is equal to that of $K^p(n,k)$.

\begin{theorem}\label{thm:IntroPPApSch}
For every prime $p$, the $\KneserPAstab^p$ problem lies in $\PPAp$.
\end{theorem}

The proofs of Theorems~\ref{thm:IntroPPAp} and~\ref{thm:IntroPPApSch} crucially rely on the mathematical arguments of Ziegler~\cite{Ziegler02} and Meunier~\cite{Meunier11}, which imply the totality of the studied problems.
We verify that their arguments can be transformed into efficient reductions to a computational problem associated with a $\Z_p$-variant of Tucker's lemma, which was shown to lie in the complexity class $\PPAp$ by Filos-Ratsikas et al.~\cite{FHSZ21}.

Finally, we apply Theorems~\ref{thm:IntroPPAp} and~\ref{thm:IntroPPApSch} to derive limitations on the complexity of variants of the $\KneserP^r$ problem, restricted to colorings with a bounded number of colors (see Theorem~\ref{thm:limit_Kneser} and Corollary~\ref{cor:Stab-PPAp}). In particular, we show that the problem of finding a monochromatic edge in a graph $S(n,k)$ given a coloring of its vertices with $\lfloor n/2 \rfloor -2k+1$ colors, lies in the complexity class $\PPAthree$ (see Corollary~\ref{cor:SchriverPPA3}). As mentioned earlier, it was shown in~\cite{Haviv23} that the latter is efficiently reducible to the $\KneserP$ problem. It thus follows that, unless $\PPA \subseteq \PPAthree$, this reduction cannot yield a $\PPA$-hardness result for the $\KneserP$ problem. We note that it is common to believe that the classes $\PPAp$ for primes $p$ do not contain each other, and that an unconditional separation between their black-box versions was provided in~\cite{Hollender21,GoosKSZ20}.

\subsection{Outline}

The rest of the paper is organized as follows.
In Section~\ref{sec:preliminaries}, we gather some definitions that will be used throughout the paper.
In Section~\ref{sec:math_proof}, we present our novel proof of Theorem~\ref{thm:IntroKriz}.
In Section~\ref{sec:comput_proof}, we study this proof from a computational perspective and prove Theorems~\ref{thm:IntroKneserToCH} and~\ref{thm:IntroKneserToCD} in generalized forms.
Finally, in Section~\ref{sec:PPAp}, we prove a generalized form of Theorem~\ref{thm:IntroPPAp}. We further prove there Theorem~\ref{thm:IntroPPApSch} and obtain some limitations on the complexity of the $\KneserP^r$ problem.

\section{Preliminaries}\label{sec:preliminaries}

For an integer $n$, let $[n] = \{1,2,\ldots,n\}$.
Throughout the paper, we identify the subsets of $[n]$ with their characteristic vectors in $\{0,1\}^n$.
For integers $n$ and $k$, let $\binom{[n]}{k}$ denote the family of $k$-subsets of $[n]$.
A subset of $[n]$ is called {\em stable} if it does not include two consecutive elements modulo $n$, equivalently, it forms an independent set in the cycle on the vertex set $[n]$ with the natural order along the cycle. Let $\binom{[n]}{k}_\TwoStab$ denote the family of stable $k$-subsets of $[n]$.
A subset of $[n]$ is called {\em almost stable} if it does not include two consecutive elements, equivalently, it forms an independent set in the path on the vertex set $[n]$ with the natural order along the path. Let $\binom{[n]}{k}_{\AlmostTwoStab}$ denote the family of almost stable $k$-subsets of $[n]$. Note that $\binom{[n]}{k}_\TwoStab \subseteq \binom{[n]}{k}_{\AlmostTwoStab} \subseteq \binom{[n]}{k}$.

The family of Kneser hypergraphs is defined as follows.
\begin{definition}[Kneser Hypergraphs]
For an integer $r \geq 2$ and a set family $\calF$, the $r$-uniform Kneser hypergraph $K^r(\calF)$ is the hypergraph on the vertex set $\calF$, whose hyperedges are all the $r$-subsets of $\calF$ whose members are pairwise disjoint. For integers $n$ and $k$ with $n \geq r \cdot k$, let $K^r(n,k)$, $K^r(n,k)_{\AlmostTwoStab}$, and $K^r(n,k)_\TwoStab$, respectively, denote the hypergraphs $K^r(\binom{[n]}{k})$, $K^r(\binom{[n]}{k}_{\AlmostTwoStab})$, and $K^r(\binom{[n]}{k}_{\TwoStab})$.
When $r=2$, the superscript $r$ may be omitted.
\end{definition}

For an integer $t \geq 1$, a hypergraph $H = (V,E)$ is said to be {\em $t$-colorable} if it admits a proper $t$-coloring, that is, a coloring of the vertices of $V$ with $t$ colors such that no hyperedge of $E$ is monochromatic. The {\em chromatic number} of $H$, denoted by $\chi(H)$, is the smallest integer $t$ for which $H$ is $t$-colorable.
It is known (see~\cite{AlonFL86,Meunier11,Frick20}) that for all integers $r \geq 2$, $k$, and $n \geq r \cdot k$, it holds that
\begin{eqnarray}\label{eq:Chromatic}
\chi(K^r(n,k)) = \chi(K^r(n,k)_{\AlmostTwoStab}) = \chi(K^r(n,k)_\TwoStab) = \Big \lceil \frac{n-r(k-1)}{r-1} \Big \rceil.
\end{eqnarray}

Theorem~\ref{thm:IntroKriz}, proved by \kriz~\cite{Kriz92}, relates the chromatic number of Kneser hypergraphs to the notion of {\em colorability defect}, defined as follows.

\begin{definition}[Colorability Defect]\label{def:color_defect}
For an integer $r \geq 2$ and a family $\calF$ of non-empty subsets of a set $X$, the {\em $r$-colorability defect} of $\calF$, denoted by $\cd_r(\calF)$, is the smallest size of a set $Y \subseteq X$, such that the sub-hypergraph of $(X,\calF)$ induced by $X \setminus Y$ is $r$-colorable, that is,
\[ {\cd}_r(\calF) = \min \big \{ |Y| ~\mid~ (X \setminus Y, \{e \in \calF \mid e \cap Y = \emptyset \}) \mbox{~is~$r$-colorable} \big \}.\]
\end{definition}

The following lemma, given in~\cite[Lemma~3.2]{Ziegler02}, determines the $r$-colorability defect of $\binom{[n]}{k}$.
\begin{lemma}[{\cite{Ziegler02}}]\label{lemma:cd_Kneser}
For integers $n, k, r$ with $n \geq r \cdot k$, it holds that $\cd_r(\binom{[n]}{k}) = n-r(k-1)$.
\end{lemma}

We next state the Consensus Division theorem due to Filos-Ratsikas et al.~\cite[Theorem~6.5]{FHSZ21}.
In what follows, let $\calB([0,1])$ denote the set of Lebesgue-measurable subsets of the interval $[0,1]$, and let $\mu: \calB([0,1]) \rightarrow [0,1]$ denote the Lebesgue measure on $[0,1]$. In addition, let $\triangle$ stand for the symmetric difference of sets, defined by $E_1 \triangle E_2 = (E_1 \setminus E_2) \cup (E_2 \setminus E_1)$.

\begin{theorem}[Consensus Division Theorem~{\cite{FHSZ21}}]\label{thm:ConDiv}
Let $p$ be a prime, and let $m \geq 1$ be an integer.
Let $v_1, \ldots, v_m : \calB([0,1]) \rightarrow \R$ be functions such that for each $i \in [m]$, $v_i$ satisfies the following continuity condition: for any $\eps >0$ there exists $\delta >0$ such that for all $E_1,E_2 \in \calB([0,1])$ with $\mu(E_1 \triangle E_2) \leq \delta$, it holds that $|v_i(E_1)-v_i(E_2)| \leq \eps$.
Then, there exists a consensus-$p$-division, that is, it is possible to partition the unit interval into $p$ (not necessarily connected) pieces $A_1, \ldots, A_p$ using at most $(p-1) \cdot m$ cuts, such that $v_i(A_{t}) = v_i(A_{t'})$ for all $i \in [m]$ and $t,t' \in [p]$.
\end{theorem}

As for the continuity property, we will sometimes consider the stronger notion of Lipschitz-continuity.
For $L \geq 0$, a function $v: \calB([0,1]) \rightarrow \R$ is said to be {\em $L$-Lipschitz-continuous}, if for all $E_1,E_2 \in \calB([0,1])$, it holds that $|v(E_1)-v(E_2)| \leq L \cdot \mu(E_1 \triangle E_2)$.

\section{The Chromatic Number of Kneser Hypergraphs}\label{sec:math_proof}

In this section, we present our new proof of \kriz's lower bound on the chromatic number of Kneser hypergraphs~\cite{Kriz92}, stated as Theorem~\ref{thm:IntroKriz}.
The proof directly applies the Consensus Division theorem, given in Theorem~\ref{thm:ConDiv}, and thus bypasses the topological notions used in~\cite{Kriz92} (see also~\cite{Matousek00,Ziegler02}).
Note that we focus here on the mathematical proof rather than on its computational aspects, which will be explored in the next section.
As is usual for results of this type, we first consider the case where the uniformity $r$ of the hypergraphs is a prime number.

\begin{theorem}\label{thm:KneserChrom_p}
For every prime $p$ and for every family $\calF$ of non-empty sets,
\[ \chi(K^p(\calF)) \geq \Big \lceil \frac{\cd_p(\calF)}{p-1} \Big \rceil.\]
\end{theorem}

\begin{proof}
Let $p$ be a prime, and let $\calF$ be a family of non-empty sets. Assume without loss of generality that all members of $\calF$ are subsets of $[n]$ for some integer $n$.
Let $m$ be an integer satisfying $m < \frac{\cd_p(\calF)}{p-1}$, and suppose for the sake of contradiction that there exists a proper coloring $c: \calF \rightarrow [m]$ of the $p$-uniform Kneser hypergraph $K^p(\calF)$.

For each $i \in [m]$, let $\widetilde{u}_i:\{0,1\}^n \rightarrow \{0,1\}$ denote the indicator function that determines for every subset of $[n]$ whether it contains a set of $\calF$ whose color according to $c$ is $i$, that is, for every $D \subseteq [n]$, define $\widetilde{u}_i(D) = 1$ if there exists a set $B \in \calF$ satisfying $B \subseteq D$ and $c(B)=i$, and define $\widetilde{u}_i(D)=0$ otherwise.
Recall that we identify the subsets of $[n]$ with their characteristic vectors in $\{0,1\}^n$.
Note that the function $\widetilde{u}_i$ is monotone with respect to inclusion.

For each $i \in [m]$, consider the extension $u_i : [0,1]^n \rightarrow [0,1]$ of $\widetilde{u}_i$ that maps any vector $x \in [0,1]^n$ to the largest value $a \in [0,1]$ such that the set $\{ j \in [n] \mid  x_j \geq a\}$ contains a set of $\calF$ colored $i$ by $c$ if such a value exists, and to $0$ otherwise.
Equivalently, for any $x \in [0,1]^n$, let $\pi$ be a permutation of $[n]$ with $x_{\pi(1)} \leq x_{\pi(2)} \leq \cdots \leq x_{\pi(n)}$, and define
$u_i(x) = x_{\pi(j)}$ for the largest integer $j$ satisfying $\widetilde{u}_i(\{\pi(j), \pi(j+1),\ldots,\pi(n)\})=1$ if such a $j$ exists, and $u_i(x) = 0$ otherwise.
Observe that, under the convention $x_{\pi(0)}=0$, the value of $u_i(x)$ can be written as
\[u_i(x) = \sum_{j=1}^{n}{(x_{\pi(j)} - x_{\pi(j-1)}) \cdot \widetilde{u}_i(\{\pi(j), \pi(j+1),\ldots,\pi(n)\})}.\]
Notice that the function $u_i$ is an extension of the function $\widetilde{u}_i$.
Notice further that $u_i$ is monotone, in the sense that the value of $u_i(x)$ does not decrease when the value of some entry of $x$ increases.
Finally, notice that changing some entry of $x$ by $\eps$ changes $u_i(x)$ by at most $\eps$, and observe that this implies that the function $u_i$ is continuous.

Now, for each $j \in [n]$, consider the open sub-interval $I_j = (\frac{j-1}{n},\frac{j}{n})$, and associate it with the element $j$.
For each $i \in [m]$, let $v_i : \calB([0,1]) \rightarrow [0,1]$ be the function defined as follows.
For every $E \in \calB([0,1])$, let $x^E \in [0,1]^n$ denote the vector that consists of the normalized Lebesgue measures of $E$ on the sub-intervals $I_1, \ldots, I_n$, that is, $x^E_j = n \cdot \mu(E \cap I_j)$ for all $j \in [n]$, and define $v_i(E) = u_i(x^E)$.
Note that the function $v_i$ is the composition of the function $u_i$ with the function that maps any set $E \in \calB([0,1])$ to the vector $x^E$.
Since the function $u_i$ is monotone and continuous, it follows that so is $v_i$

Applying Theorem~\ref{thm:ConDiv}, we obtain that there exists a consensus-$p$-division of $v_1, \ldots, v_m$, that is, a partition of the unit interval into $p$ pieces $A_1, \ldots, A_p$ using at most $(p-1) \cdot m$ cuts, such that $v_i(A_{t}) = v_i(A_{t'})$ for all $i \in [m]$ and $t,t' \in [p]$.
Let $Y$ denote the set of indices $j \in [n]$ for which the open sub-interval $I_j$ includes a cut, and notice that $|Y| \leq (p-1) \cdot m < \cd_p(\calF)$.
Thus, every sub-interval $I_j$ with $j\in [n] \setminus Y$ is fully contained in one of the pieces $A_1, \ldots, A_p$.
Consider the coloring that assigns to every $j \in [n]\setminus Y$ the index $t \in [p]$ of the piece $A_t$ that contains $I_j$.
By $|Y| < \cd_p(\calF)$, the hypergraph $([n] \setminus Y,\{e \in \calF \mid e \cap Y = \emptyset\})$ is not $p$-colorable.
This implies that for some $t \in [p]$, there exists a set $B \in \calF$ all of whose elements share the color $t$, and thus, for all $j \in B$, it holds that $I_j \subseteq A_t$. Denoting by $\ell = c(B)$ the color assigned to $B$ by the given coloring $c$, it follows from the definition of $v_\ell$ that $v_\ell(\cup_{j \in B}{I_j}) = 1$.
By monotonicity, it further follows that $v_{\ell}(A_t)=1$, which yields, by Theorem~\ref{thm:ConDiv}, that $v_{\ell}(A_{t'}) = 1$ for all $t' \in [p]$.

Finally, for every $t' \in [p]$, the fact that $v_{\ell}(A_{t'}) = 1$ implies that there exists a set $B_{t'} \in \calF$ with $c(B_{t'}) = \ell$ such that $\mu(A_{t'} \cap I_j)=\tfrac{1}{n}$ for all $j \in B_{t'}$. Since the pieces $A_1, \ldots, A_p$ are pairwise disjoint, so are the sets $B_1, \ldots, B_p$. It thus follows that these sets form a monochromatic hyperedge in $K^p(\calF)$, in contradiction to the assumption that the coloring $c$ is proper. This completes the proof.
\end{proof}

It is well known that Theorem~\ref{thm:KneserChrom_p} implies Theorem~\ref{thm:IntroKriz} (see, e.g.,~\cite{Ziegler02}).
We provide the quick proof for completeness.
\begin{proof}[ of Theorem~\ref{thm:IntroKriz}]
Theorem~\ref{thm:KneserChrom_p} shows that Theorem~\ref{thm:IntroKriz} holds whenever $r$ is prime.
It therefore suffices to prove that for every pair of integers $r_1, r_2 \geq 2$, if the theorem holds for $r \in \{r_1, r_2\}$ then it holds for $r = r_1 r_2$. So suppose that it holds for $r_1$ and $r_2$. Let $\calF$ be a family of non-empty subsets of $[n]$ for an integer $n$, and let $c: \calF \rightarrow [m]$ be  a proper coloring of the hypergraph $K^{r_1 r_2}(\calF)$ for an integer $m$. Our goal is to prove that $m \geq \frac{\cd_{r_1 r_2}(\calF)}{r_1 r_2 -1}$.

Consider the family
\[ \calG = \Big \{ G \subseteq [n] ~\Big | ~{\cd}_{r_1}(\calF|_G) > m(r_1-1) \Big \},\]
where $\calF|_G = \{ F \in \calF \mid F \subseteq G \}$.
Define a coloring $c': \calG \rightarrow [m]$ as follows.
For every $G \in \calG$, let $c'(G)$ be a color of some monochromatic hyperedge in $K^{r_1}(\calF|_G)$ with respect to the coloring $c$. The existence of such a hyperedge for every $G \in \calG$ follows from our assumption that the theorem holds for $r_1$, which yields that $\chi(K^{r_1}(\calF|_G)) \geq \frac{\cd_{r_1}(\calF|_G)}{r_1-1}>m$, where the second inequality is due to the definition of $\calG$. We claim that $c'$ is a proper coloring of $K^{r_2}(\calG)$, and thus $\chi (K^{r_2}(\calG)) \leq m$. Indeed, otherwise there would exist $r_2$ pairwise disjoint sets $G_1, \ldots, G_{r_2} \in \calG$ that are assigned the same color by $c'$. Since each set $G_i$ with $i \in [r_2]$ contains $r_1$ pairwise disjoint sets of $\calF|_{G_i}$ that are assigned by $c$ the color $c'(G_i)$, this gives us $r_1 r_2$ pairwise disjoint sets of $\calF$ with the same color by $c$, contradicting the assumption that $c$ is a proper coloring of $K^{r_1 r_2}(\calF)$.

Finally, by the assumption that the theorem holds for $r_2$, we have $\chi(K^{r_2}(\calG)) \geq \frac{\cd_{r_2}(\calG)}{r_2-1}$, which implies that $\cd_{r_2}(\calG) \leq m(r_2-1)$. This means that it is possible to remove at most $m(r_2-1)$ of the elements of $[n]$ and to partition the remaining elements into $r_2$ sets $F_1, \ldots, F_{r_2}$, such that no set of $\calG$ is contained in any of them. In particular, for each $i \in [r_2]$, it holds that $F_i \notin \calG$, which implies that $\cd_{r_1}(\calF|_{F_i}) \leq m(r_1-1)$.
This means that it is possible to remove at most $m(r_1-1)$ elements from each $F_i$ and to partition the remaining elements into $r_1$ sets, such that no set of $\calF|_{F_i}$ is contained in any of them. It thus follows that one can remove at most $m(r_2-1) + r_2 \cdot m(r_1-1) = m(r_1 r_2-1)$ elements from $[n]$ and partition the remaining elements into $r_1 r_2$ sets, such that no set of $\calF$ is contained in any of them. This implies that $\cd_{r_1 r_2}(\calF) \leq m(r_1 r_2-1)$, providing the desired lower bound on $m$.
\end{proof}

We conclude this section by verifying that $\chi(K^r(n,k)) = \lceil \frac{n-r(k-1)}{r-1} \rceil$ for all $r \geq 2$.
The lower bound on $\chi(K^r(n,k))$, originally proved in~\cite{AlonFL86}, follows by combining Theorem~\ref{thm:IntroKriz} with Lemma~\ref{lemma:cd_Kneser}.
For the upper bound, which is given in~\cite{Erdos76}, put $t = \lceil \frac{n-r(k-1)}{r-1} \rceil$, let $X_1, \ldots, X_{t-1}$ be $t-1$ pairwise disjoint $(r-1)$-subsets of $[n]$, and observe that the number of elements of $[n]$ that do not belong to any of these sets is $n-(t-1)(r-1) \leq rk-1$. Consider the coloring that assigns to every $k$-subset $B$ of $[n]$ a color $i \in [t-1]$ such that $B \cap X_i \neq \emptyset$ if such an $i$ exists, and the color $t$ otherwise. It can be easily checked that this coloring of $K^r(n,k)$ is proper.

\section{Reduction from \texorpdfstring{$\KneserP^p$}{Kneser-p} to Approximate \texorpdfstring{$\pConDiv$}{Con-p-Div}}\label{sec:comput_proof}

In this section, we use the proof technique presented in Section~\ref{sec:math_proof} to obtain efficient reductions from the computational problems associated with Kneser hypergraphs to those associated with approximate consensus division. We start by formally introducing the involved computational problems and then present the reductions.

\subsection{The \texorpdfstring{$\KneserP^r$}{Kneser-r} Problem}\label{subsec:Kneser}

In order to make our results as strong as possible, we introduce a problem of finding monochromatic hyperedges in general Kneser hypergraphs, defined as follows.

\begin{tcolorbox}
\begin{definition}[The $\KneserP^r(\calF,m)$ Problem]\label{def:KneserProblem}
For a set $\calA$, let $\calF = (\calF_\alpha)_{\alpha \in \calA}$ be a sequence of set families, where for each $\alpha \in \calA$, $\calF_\alpha$ is a family of non-empty subsets of $[n_\alpha]$ for some integer $n_\alpha$, and let $m:\calA \rightarrow \N$ be a function.
In the $\KneserP^r(\calF,m)$ problem, the input consists of
\begin{itemize}\setlength\itemsep{-0.2em}
  \item an element $\alpha \in \calA$ and
  \item a Boolean circuit $C:\{0,1\}^{n_\alpha} \rightarrow [m(\alpha)]$ that represents a coloring $c:\calF_\alpha \rightarrow [m(\alpha)]$ of the sets of $\calF_\alpha$ with $m(\alpha)$ colors, in the sense that for every $B \in \calF_\alpha$, it holds that $C(B) = c(B)$.
\end{itemize}
The goal is to find a monochromatic hyperedge in $K^r(\calF_\alpha)$, that is, $r$ pairwise disjoint sets $B_1 , B_2,\ldots, B_r \in \calF_\alpha$ such that $C(B_1) = C(B_2) = \cdots = C(B_r)$.
\end{definition}
\end{tcolorbox}
\noindent
Note that for every sequence $\calF = (\calF_\alpha)_{\alpha \in \calA}$ and for every function $m: \calA \rightarrow \N$ for which it holds that $m(\alpha) < \chi(K^r(\calF_\alpha))$ for all $\alpha \in \calA$, the $\KneserP^r(\calF,m)$ problem is total.
In particular, by Theorem~\ref{thm:IntroKriz}, the $\KneserP^r(\calF,\lfloor \frac{{\cd}_r(\calF_\alpha)-1}{r-1} \rfloor )$ problem is total, where the latter is an abbreviation for the $\KneserP^r(\calF,m)$ problem with the function $m$ defined by $m(\alpha)=\lfloor \frac{{\cd}_r(\calF_\alpha)-1}{r-1} \rfloor$.

We next define a variant of the $\KneserP^r(\calF,m)$ problem with an extended access to the input coloring, referred to as {\em subset queries}. Intuitively, subset queries facilitate transforming our proof of Theorem~\ref{thm:KneserChrom_p} into an efficient reduction, as the proof applies the Consensus Division theorem to functions that determine whether a given set contains a {\em subset} assigned a specified color.

\begin{tcolorbox}
\begin{definition}[The $\KneserP^r(\calF,m)$ Problem with Subset Queries]\label{def:KneserProblemSubset}
For a set $\calA$, let $\calF = (\calF_\alpha)_{\alpha \in \calA}$ be a sequence of set families, where for each $\alpha \in \calA$, $\calF_\alpha$ is a family of non-empty subsets of $[n_\alpha]$ for some integer $n_\alpha$, and let $m:\calA \rightarrow \N$ be a function.
In the $\KneserP^r(\calF,m)$ problem {\em with subset queries}, the input consists of
\begin{itemize}\setlength\itemsep{-0.2em}
  \item an element $\alpha \in \calA$,
  \item a Boolean circuit $C:\{0,1\}^{n_\alpha} \rightarrow [m(\alpha)]$ that represents a coloring $c:\calF_\alpha \rightarrow [m(\alpha)]$ of the sets of $\calF_\alpha$ with $m(\alpha)$ colors, in the sense that for every $B \in \calF_\alpha$, it holds that $C(B) = c(B)$, and
  \item a Boolean circuit $S:\{0,1\}^{n_\alpha} \times [m(\alpha)] \rightarrow \{0,1\}$ that is supposed to allow subset queries to the coloring $c$, namely, for every set $D \subseteq [n_\alpha]$ and a color $i \in [m(\alpha)]$, it is supposed to satisfy $S(D,i)=1$ if there exists a set $B \in \calF_\alpha$ such that $B \subseteq D$ and $c(B)=i$, and $S(D,i)=0$ otherwise.
\end{itemize}
The goal is to find either a monochromatic hyperedge in $K^r(\calF_\alpha)$ or a violation of the circuit $S$, namely,
\begin{itemize}\setlength\itemsep{-0.2em}
  \item (false negative) two sets $B,D \subseteq [n_\alpha]$ and a color $i \in [m(\alpha)]$ such that $B \in \calF_\alpha$, $B \subseteq D$, $C(B)=i$, and yet $S(D,i)=0$, or
  \item (false positive) a set $D \subseteq [n_\alpha]$ and a color $i \in [m(\alpha)]$ such that $S(D,i)=1$ whereas for every set $D'$ obtained from $D$ by removing a single element it holds that $S(D',i)=0$, and in addition, either $D \notin \calF_\alpha$ or $D \in \calF_\alpha$ and $C(D) \neq i$.
\end{itemize}
\end{definition}
\end{tcolorbox}
\noindent
Note that by allowing the solutions of the $\KneserP^r(\calF,m)$ problem with subset queries to form violations of the circuit $S$, we obtain a non-promise problem.

We will be particularly interested in the following special cases of Definitions~\ref{def:KneserProblem} and~\ref{def:KneserProblemSubset}.
For any integer $r \geq 2$, let $\calA^{(r)}$ denote the set of all pairs of integers $(n,k)$ with $n \geq r \cdot k$, and let $\calF^{(r)}$, $\calF^{(r,\mathrm{\TwoStab})}$, and $\calF^{(r,\AlmostTwoStab)}$ denote, respectively, the sequences of set families defined by
\[\calF^{(r)}_{(n,k)} = \binom{[n]}{k},~~~~~\calF^{(r,\TwoStab)}_{(n,k)} = \binom{[n]}{k}_{\TwoStab},~~~~~\calF^{(r,\AlmostTwoStab)}_{(n,k)} = \binom{[n]}{k}_{\AlmostTwoStab}\]
for all $(n,k) \in \calA^{(r)}$. Consider the function $m^{(r)} : \calA^{(r)} \rightarrow \N$ defined by $m^{(r)}(n,k) = \lfloor \frac{n-r(k-1)-1}{r-1} \rfloor$, which by~\eqref{eq:Chromatic}, satisfies $m^{(r)}(n,k) = \chi ( K^r(n,k))-1$. We define the three problems $\KneserP^r$, $\KneserP^r_{\TwoStab}$, and $\KneserP^r_{\AlmostTwoStab}$, respectively, to be $\KneserP^r(\calF^{(r)},m^{(r)})$, $\KneserP^r(\calF^{(r,\TwoStab)},m^{(r)})$, and $\KneserP^r(\calF^{(r,\AlmostTwoStab)},m^{(r)})$.
For any function $m:\calA^{(r)} \rightarrow \N$, we let $\KneserPstab^r(n,k,m)$ denote the $\KneserP^r(\calF^{(r,\TwoStab)},m)$ problem.
When $r=2$, the superscript $r$ may be omitted.
For any function $m:\calA^{(2)} \rightarrow \N$, we let $\KneserP(n,k,m)$ and $\SchrijverP(n,k,m)$ denote, respectively, the problems $\KneserP(\calF^{(2)},m)$ and $\KneserP(\calF^{(2,\TwoStab)},m)$.

We will use the notion of polynomially computable sequences of set families, defined below.
\begin{definition}\label{def:computable_seq}
For a set $\calA$, let $\calF = (\calF_\alpha)_{\alpha \in \calA}$ be a sequence of set families, where for each $\alpha \in \calA$, $\calF_\alpha$ is a family of non-empty subsets of $[n_\alpha]$ for some integer $n_\alpha$.
The sequence $\calF$ is {\em polynomially computable} if there exist polynomials $q_1,q_2$ such that
\begin{enumerate}
  \item\label{itm:weak1} there exists an algorithm that given an element $\alpha \in \calA$ and a set $B \subseteq [n_\alpha]$ runs in time $q_1(n_\alpha)$ and determines whether $B \in \calF_\alpha$, and
  \item\label{itm:weak2} there exists an algorithm that given an element $\alpha \in \calA$ and a set $D \subseteq [n_\alpha]$ runs in time $q_2(n_\alpha)$, returns a subset of $D$ that belongs to the family $\calF_\alpha$ if such a subset exists, and declares that no such subset exists otherwise.
\end{enumerate}
\end{definition}

The following lemma gives simple examples of polynomially computable sequences.
\begin{lemma}\label{lemma:weaklyFnk}
For every integer $r \geq 2$, the sequence $\calF^{(r)}$ is polynomially computable.
\end{lemma}

\begin{proof}
Fix an integer $r \geq 2$. Recall that $\calF^{(r)}_{(n,k)} = \binom{[n]}{k}$ for all $(n,k) \in \calA^{(r)}$.
We show that $\calF^{(r)}$ satisfies the conditions of Definition~\ref{def:computable_seq}.
For Item~\ref{itm:weak1}, consider the algorithm that given a pair $(n,k) \in \calA^{(r)}$ and a set $B \subseteq [n]$, checks whether $|B|=k$.
For Item~\ref{itm:weak2}, consider the algorithm that given a pair $(n,k) \in \calA^{(r)}$ and a set $D \subseteq [n]$, checks whether $|D|\geq k$, and if so, returns an arbitrary subset of $D$ of size $k$.
Clearly, these algorithms can be implemented in time polynomial in $n$, hence $\calF^{(r)}$ is polynomially computable.
\end{proof}

We conclude this section with two lemmas on the $\KneserP^r(\calF,m)$ problem with subset queries.

\begin{lemma}\label{lemma:subset_in_F}
Let $\calF = (\calF_\alpha)_{\alpha \in \calA}$ be a polynomially computable sequence of set families, where for each $\alpha \in \calA$, $\calF_\alpha$ is a family of non-empty subsets of $[n_\alpha]$ for some integer $n_\alpha$, and let $m:\calA \rightarrow \N$ be a function.
Given an instance $(\alpha, C, S)$ of the $\KneserP^r(\calF,m)$ problem with subset queries and given a set $D \subseteq [n_\alpha]$ and a color $i \in [m(\alpha)]$ such that $S(D,i)=1$, it is possible to find in polynomial time either a set $B \in \calF_\alpha$ such that $B \subseteq D$ and $C(B)=i$ or a violation of the circuit $S$.
\end{lemma}

\begin{proof}
Suppose that we are given an instance $(\alpha, C, S)$, a set $D \subseteq [n_\alpha]$, and a color $i \in [m(\alpha)]$, as in the statement of the lemma.
Consider the algorithm that first checks if $D \in \calF_\alpha$ and $C(D)=i$. If so, the algorithm simply returns $D$, which satisfies the required properties. Otherwise, the algorithm goes over all the subsets $D' \subseteq D$ obtained from $D$ by removing a single element and evaluates the circuit $S$ at the pairs $(D',i)$. If all those sets $D'$ satisfy $S(D',i)=0$, then the set $D$ and the color $i$ form a false positive violation of the circuit $S$, and the algorithm returns this violation. Otherwise, some $D'$ satisfies $S(D',i)=1$, and the algorithm proceeds by repeating the process where $D$ is replaced by such $D'$. Note that the number of iterations does not exceed $n_\alpha$.
The total running time is polynomial, because $\calF$ is polynomially computable, which allows us to efficiently check membership in $\calF_\alpha$, and because the circuits $C$ and $S$ can be evaluated in polynomial time.
\end{proof}

\begin{lemma}\label{lemma:viloationSubset}
Let $\calF = (\calF_\alpha)_{\alpha \in \calA}$ be a polynomially computable sequence of set families, where for each $\alpha \in \calA$, $\calF_\alpha$ is a family of non-empty subsets of $[n_\alpha]$ for some integer $n_\alpha$, and let $m:\calA \rightarrow \N$ be a function.
Given an instance $(\alpha, C, S)$ of the $\KneserP^r(\calF,m)$ problem with subset queries and given two sets $D_1, D_2 \subseteq [n_\alpha]$ and a color $i \in [m(\alpha)]$ such that $D_1 \subseteq D_2$, $S(D_1,i)=1$, and $S(D_2,i)=0$, it is possible to find in polynomial time a violation of the circuit $S$.
\end{lemma}

\begin{proof}
Suppose that we are given an instance $(\alpha, C, S)$, sets $D_1, D_2 \subseteq [n_\alpha]$, and a color $i \in [m(\alpha)]$, as in the statement of the lemma.
Consider the algorithm that first checks whether $D_1 \in \calF_\alpha$ and $C(D_1)=i$. If so, the sets $D_1,D_2$ and the color $i$ form a false negative violation of $S$, and the algorithm returns this violation.
Otherwise, $D_1 \notin \calF_\alpha$ or $C(D_1) \neq i$.
In this case, by running the algorithm from Lemma~\ref{lemma:subset_in_F} on $(D_1,i)$, we get either a set $B \in \calF_\alpha$ such that $B \subseteq D_1$ and $C(B)=i$ or a violation of $S$. Notice that in the former case, the sets $B,D_2$ and the color $i$ form a false negative violation of $S$, hence in both cases, the algorithm can return a violation of $S$, and we are done.
Note that the running time of the algorithm is polynomial.
\end{proof}

\subsection{The \texorpdfstring{$\pConDiv$}{Con-p-Div} Problem}\label{subsec:CD}

We present now the formal definition of the approximate $\pConDiv$ problem for any prime $p$. The definition essentially extends the one given for $p=2$ in~\cite[Appendix~B, Definition~8]{DeligkasFH22}. One difference between the definitions is that we require a strict inequality in the condition~\eqref{eq:starict_def} below. For this reason, we denote the problem by $\ConDiv{\eps}{p}$ and deviate from the notation $\eps \textsc{-} \pConDiv$ used in the literature.
Note that we do not make any assumptions on the input valuation functions, and therefore allow solutions that demonstrate violations of their expected properties.

\begin{tcolorbox}
\begin{definition}[The $\pConDiv$ Problem]\label{def:ConDiv}
For an $\eps \in (0,1]$, a prime $p$, and a fixed polynomial $q$, the $\ConDiv{\eps}{p}$ problem on normalized monotone functions is defined as follows.
The input consists of
\begin{itemize}\setlength\itemsep{-0.2em}
  \item a Lipschitz parameter $L \geq 0$ and
  \item $m$ Turing machines $v_1, \ldots, v_m$ that are supposed to compute $L$-Lipschitz-continuous normalized monotone valuation functions in $\calB([0,1]) \rightarrow [0,1]$.
\end{itemize}
The goal is to find either a partition of $[0,1]$ into $p$ pieces $A_1, \ldots, A_p$ using at most $(p-1) \cdot m$ cuts, such that
  \begin{eqnarray}\label{eq:starict_def}
  |v_i(A_t)-v_i(A_{t'})| < \eps~~~~~\text{for all $i \in [m]$ and $t,t' \in [p]$,}
  \end{eqnarray}
or a violation of some valuation function $v_i$, namely,
\begin{itemize}\setlength\itemsep{-0.2em}
  \item a violation of the normalization of $v_i$, or
  \item a violation of the running time of $v_i$, i.e., an input $E$ on which $v_i$ does not terminate within $q(|E|+|v_i|)$ steps, or
  \item a violation of the monotonicity of $v_i$, or
  \item a violation of the $L$-Lipschitz-continuity of $v_i$.
\end{itemize}
When $p=2$, we refer to the $\ConDiv{\eps}{2}$ problem as $\ConHalvi{\eps}$.
\end{definition}
\end{tcolorbox}
\noindent
Theorem~\ref{thm:ConDiv} implies that for any $\eps \in (0,1]$ and for every prime $p$, the $\ConDiv{\eps}{p}$ problem is total and thus lies in $\TFNP$.

\subsection{From \texorpdfstring{$\KneserP$}{Kneser-p} to \texorpdfstring{$\ConHalvi{1}$}{Con-Halving[<1]}}

We consider now the case $p=2$ and present our reduction from the general $\KneserP(\calF, m)$ problem with subset queries, where $m$ is smaller by one than the bound given by Theorem~\ref{thm:IntroKriz}, to the $\ConHalvi{\eps}$ problem with $\eps=1$. As a consequence, we will obtain Theorem~\ref{thm:IntroKneserToCH}.

\begin{theorem}\label{thm:Kneser->CH1}
Let $\calF = (\calF_\alpha)_{\alpha \in \calA}$ be a polynomially computable sequence of set families.
Then, there exists a polynomial-time reduction from the $\KneserP(\calF, \cd_2(\calF_\alpha)-1)$ problem with subset queries to the $\ConHalvi{1}$ problem on normalized monotone functions.
\end{theorem}

\begin{proof}
For a polynomially computable sequence of set families $\calF = (\calF_\alpha)_{\alpha \in \calA}$, consider an instance of the $\KneserP(\calF, \cd_2(\calF_\alpha)-1)$ problem with subset queries, that is, an element $\alpha \in \calA$ and Boolean circuits $C:\{0,1\}^{n} \rightarrow [m]$ and $S:\{0,1\}^{n} \times [m] \rightarrow \{0,1\}$ that are supposed to represent a coloring $c:\calF_\alpha \rightarrow [m]$ of the vertices of $K(\calF_\alpha)$ where $n = n_\alpha$ and $m = \cd_2(\calF_\alpha)-1$ (see Definition~\ref{def:KneserProblemSubset}). Consider the reduction that given such an input returns $m$ Turing machines that compute the valuation functions $v_1, \ldots, v_m : \calB([0,1]) \rightarrow [0,1]$ defined as follows. For each $i \in [m]$ and any $E \in \calB([0,1])$, let $x^E \in [0,1]^n$ denote the vector defined by $x^E_j = n \cdot \mu(E \cap I_j)$ where $I_j = (\frac{j-1}{n},\frac{j}{n})$ for all $j \in [n]$, let $\pi$ be a permutation of $[n]$ such that $x^E_{\pi(1)} \leq x^E_{\pi(2)} \leq \cdots \leq x^E_{\pi(n)}$, and define
\begin{eqnarray}\label{eq:v_i(C)}
v_i(E) = \sum_{j=1}^{n}{(x^E_{\pi(j)} - x^E_{\pi(j-1)}) \cdot S(\{\pi(j), \pi(j+1),\ldots,\pi(n)\},i)},
\end{eqnarray}
where $x^E_{\pi(0)}=0$.
Notice that $v_i(E)$ is a combination of $0,1$ values with non-negative coefficients whose sum is at most $1$. This implies that $v_i(E) \in [0,1]$ for all $E \in \calB([0,1])$, that is, $v_i$ is normalized.

We claim that given the Boolean circuit $S$, it is possible to produce in polynomial time efficient Turing machines that compute the functions $v_1,\ldots,v_m$. Indeed, for each $i \in [m]$, consider the Turing machine that given any $E \in \calB([0,1])$, represented as a union of sub-intervals of $[0,1]$ with rational endpoints, computes the vector $x^E$ whose values are the normalized Lebesgue measures of $E$ on the intervals $I_1, \ldots, I_n$ and finds a permutation $\pi$ of $[n]$ that corresponds to a non-decreasing order of these values. Then, the machine determines the value of $v_i(E)$ given in~\eqref{eq:v_i(C)} using $n$ evaluations of the Boolean circuit $S$. Observe that the running time of such a machine is polynomial in the size of the representation of $E$ and in the size of the circuit $S$, which is not larger than the size of the encoding of the machine itself. We further define the Lipschitz parameter to be $L=n$. This gives us an instance of the $\ConHalvi{1}$ problem on normalized monotone functions.

Before proving the correctness of the reduction, let us add some notation.
For a set $E \in \calB([0,1])$ and a color $i \in [m]$, let $a^E_i$ denote the largest value $a \in [0,1]$ such that $S(\{j \in [n] \mid x^E_j \geq a\},i)=1$ if such an $a$ exists, and $a^E_i = 0$ otherwise.
Observe, using~\eqref{eq:v_i(C)}, that if the answers of the circuit $S$ are consistent with some coloring as described in Definition~\ref{def:KneserProblemSubset}, then for all $E \in \calB([0,1])$ and $i \in [m]$, it holds that $v_i(E) = a^E_i$. The following lemma shows that given a set $E$ and a color $i$ for which this equality does not hold, it is possible to efficiently find a violation of the circuit $S$.

\begin{lemma}\label{lemma:violation}
Given a set $E \in \calB([0,1])$ and a color $i \in [m]$ such that $v_i(E) \neq a^E_i$, it is possible to find in polynomial time a violation of the circuit $S$.
\end{lemma}

\begin{proof}
Consider the algorithm that given a set $E \in \calB([0,1])$ and a color $i \in [m]$ with $v_i(E) \neq a^E_i$, computes the vector $x^E$ and a permutation $\pi$ of $[n]$ that corresponds to a non-decreasing order of the values of $x^E$. Then, for each $j \in [n]$, the algorithm evaluates the circuit $S$ at the pair $(T_j,i)$, where $T_j = \{\pi(j), \pi(j+1), \ldots, \pi(n)\}$.
By $v_i(E) \neq a^E_i$, using~\eqref{eq:v_i(C)}, it follows that some $j$ satisfies that $S(T_j,i)=0$ and $S(T_{j+1},i)=1$. Such a $j$ can be found efficiently, and by Lemma~\ref{lemma:viloationSubset}, the sets $T_j,T_{j+1}$ (which satisfy $T_{j+1} \subseteq T_j$) and the color $i$ can be used to efficiently find a violation of $S$.
\end{proof}

We prove now the correctness of the reduction, namely, that given a solution for the produced $\ConHalvi{1}$ instance, it is possible to find in polynomial time a solution for the given instance of $\KneserP(\calF, \cd_2(\calF_\alpha)-1)$ with subset queries.
Suppose first that the given solution is a violation of some valuation function $v_i$ with $i \in [m]$.
By construction, it is impossible to violate the normalization of $v_i$ or the running time of the Turing machine that computes it, hence the solution violates either the monotonicity or the $L$-Lipschitz-continuity of $v_i$.

If the solution violates the monotonicity of $v_i$, then we are given two sets $E_1,E_2 \in \calB([0,1])$ such that $E_1 \subseteq E_2$ and $v_i(E_2) < v_i(E_1)$. We may assume that $v_i(E_1) = a^{E_1}_i$ and $v_i(E_2) = a^{E_2}_i$, and thus $a^{E_2}_i < a^{E_1}_i$, as otherwise, by Lemma~\ref{lemma:violation}, it is possible to efficiently find a violation of the circuit $S$.
Put $D_1 = \{ j \in [n] \mid x^{E_1}_j \geq a^{E_1}_i \}$, and notice that the definition of $a^{E_1}_i$ implies that $S(D_1,i)=1$.
Put $D_2 = \{ j \in [n] \mid x^{E_2}_j \geq a^{E_1}_i \}$, and use the fact that $E_1 \subseteq E_2$ to obtain that $x^{E_1}_j \leq x^{E_2}_j$ for all $j \in [n]$, which implies that $D_1 \subseteq D_2$. However, using $a^{E_2}_i < a^{E_1}_i$, it follows that $S(D_2,i)=0$. By Lemma~\ref{lemma:viloationSubset}, the sets $D_1, D_2$ and the color $i$ can be used to efficiently find a violation of $S$.

If the solution violates the $L$-Lipschitz-continuity of $v_i$ for $L=n$, then we are given two sets $E_1,E_2 \in \calB([0,1])$ such that for some $\delta > 0$, $\mu(E_1 \triangle E_2) = \delta$ and $|v_i(E_1)-v_i(E_2)| > n \cdot \delta$. Assume without loss of generality that $v_i(E_1) \geq v_i(E_2)$, which implies that $v_i(E_1) > v_i(E_2)+n \cdot \delta$.
Put $E_3 = E_1 \cup E_2$. We may assume that $v_i(E_3) \geq v_i(E_1)$, and thus $v_i(E_3) > v_i(E_2)+n \cdot \delta$, as otherwise the sets $E_1$ and $E_3$ violate the monotonicity of $v_i$, and this violation can be handled as before.
We may also assume, by Lemma~\ref{lemma:violation}, that $v_i(E_2) = a^{E_2}_i$ and $v_i(E_3) = a^{E_3}_i$, implying that $a_i^{E_3} > a_i^{E_2}+n \cdot \delta$.
Put $D_3 = \{ j \in [n] \mid x^{E_3}_j \geq a^{E_3}_i \}$, and notice that the definition of $a^{E_3}_i$ implies that $S(D_3,i)=1$.
Now, put $D_2 = \{ j \in [n] \mid x^{E_2}_j \geq a^{E_3}_i - n \cdot \delta \}$. By $\mu(E_1 \triangle E_2) = \delta$, it follows that $x^{E_3}_j \leq x^{E_2}_j + n \cdot \delta$ for all $j \in [n]$, which implies that $D_3 \subseteq D_2$. However, using $a^{E_2}_i < a^{E_3}_i - n \cdot \delta$, it follows that $S(D_2,i)=0$. By Lemma~\ref{lemma:viloationSubset}, the sets $D_2, D_3$ and the color $i$ can be used to efficiently find a violation of $S$.

It remains to consider a solution which forms a partition of the unit interval into two pieces $A_1$ and $A_2$ with at most $m$ cuts, such that $|v_i(A_{1}) - v_i(A_{2})|<1$ for all $i \in [m]$.
Given such a solution, it is possible to efficiently construct the set $Y$ of indices $j \in [n]$ for which the sub-interval $I_j$ contains a cut.
Every sub-interval $I_j$ with $j\in [n] \setminus Y$ is fully contained in either $A_1$ or $A_2$.
For $t \in [2]$, put $M_t = \{ j \in [n] \setminus Y \mid I_j \subseteq A_t\}$.
By $|Y| \leq m < \cd_2(\calF_\alpha)$, it follows that for some $t_1 \in [2]$, there exists a set $B_1 \in \calF_\alpha$ such that $B_1 \subseteq M_{t_1}$.
Moreover, such a set $B_1$ can be found efficiently, because the sequence $\calF$ is polynomially computable (see Definition~\ref{def:computable_seq}, Item~\ref{itm:weak2}).
Note that for each $j \in B_1$, it holds that $\mu(A_{t_1} \cap I_j) = \tfrac{1}{n}$.

Let $\ell = c(B_1)$ denote the color assigned to the set $B_1$ by the input coloring $c$. This color can be efficiently determined using the circuit $C$. It may be assumed that $S(B_1,\ell)=1$, as otherwise $B_1$ and $\ell$ provide a false negative violation of the circuit $S$.
Put $E = \cup_{j \in B_1}{I_j}$. By Lemma~\ref{lemma:violation}, we may assume that $v_\ell(E)=a^E_\ell=1$.
Since $E \subseteq A_{t_1}$, if $v_\ell(A_{t_1})<1$ we obtain a violation of the monotonicity of $v_\ell$, which can be handled as before.
Otherwise, it holds that $v_\ell(A_{t_1})=1$.
Our assumption on the given solution implies that $|v_\ell(A_1)-v_\ell(A_2)|<1$, so letting $t_2$ denote the element of $[2] \setminus \{t_1\}$, it follows that $v_{\ell}(A_{t_2})>0$.
Again, by Lemma~\ref{lemma:violation}, we may assume that $v_\ell(A_{t_2})=a^{A_{t_2}}_\ell >0$.
This yields that the set $D = \{ j \in [n] \mid x^{A_{t_2}}_j \geq a^{A_{t_2}}_\ell\}$ satisfies $S(D,\ell)=1$. By Lemma~\ref{lemma:subset_in_F}, it is possible to efficiently find either a violation of $S$ or a set $B_2 \in \calF_\alpha$ such that $B_2 \subseteq D$ and $C(B_2)=\ell$.
We finally observe that the sets $B_1$ and $B_2$ are disjoint.
Indeed, for each $j \in [n]$, if $j \in B_1$ then $\mu(A_{t_1} \cap I_j) = \tfrac{1}{n}$ and if $j \in B_2$ then $\mu(A_{t_2} \cap I_j) \geq \tfrac{1}{n} \cdot a^{A_{t_2}}_\ell >0$, but these two conditions cannot hold simultaneously, because $A_1$ and $A_2$ form a partition of $[0,1]$, hence $\mu(A_{t_1} \cap I_j)+\mu(A_{t_2} \cap I_j) = \mu(I_j) = \tfrac{1}{n}$.
This implies that the sets $B_1$ and $B_2$ form a monochromatic edge in $K(\calF_\alpha)$ and thus a solution for the given instance of $\KneserP(\calF, \cd_2(\calF_\alpha)-1)$ with subset queries, completing the proof.
\end{proof}

We are ready to derive Theorem~\ref{thm:IntroKneserToCH}.
Recall that $\calF^{(2)}$ stands for the sequence of the set families $\calF^{(2)}_{(n,k)} = \binom{[n]}{k}$ for pairs of integers $(n,k)$ with $n \geq 2k$.

\begin{proof}[ of Theorem~\ref{thm:IntroKneserToCH}]
By Lemma~\ref{lemma:weaklyFnk}, the sequence $\calF^{(2)}$ is polynomially computable.
This allows us to apply Theorem~\ref{thm:Kneser->CH1}, which yields that there exists a polynomial-time reduction from $\KneserP(\calF^{(2)}, \cd_2(\calF^{(2)}_{(n,k)})-1)$ with subset queries to $\ConHalvi{1}$ on normalized monotone functions.
By Lemma~\ref{lemma:cd_Kneser}, it holds that $\cd_2(\calF^{(2)}_{(n,k)}) = n-2k+2$. It thus follows that the $\KneserP(\calF^{(2)}, \cd_2(\calF^{(2)}_{(n,k)})-1)$ problem coincides with the $\KneserP$ problem, and we are done.
\end{proof}

\subsection{From \texorpdfstring{$\KneserP^p$}{Kneser-p} to \texorpdfstring{$\ConDiv{\tfrac{1}{2}}{p}$}{Con-p-Div[<1/2]}}

We next present a reduction, for any prime $p \geq 3$, from the $\KneserP^p(\calF, m)$ problem with subset queries, where $m$ is smaller than the bound given by Theorem~\ref{thm:IntroKriz}, to the $\ConDiv{\eps}{p}$ problem with $\eps=\tfrac{1}{2}$. As a consequence, we will obtain Theorem~\ref{thm:IntroKneserToCD}.

\begin{theorem}\label{thm:Kneser->CD}
Let $p \geq 3$ be a prime, and let $\calF = (\calF_\alpha)_{\alpha \in \calA}$ be a polynomially computable sequence of set families.
Then, there exists a polynomial-time reduction from the $\KneserP^p(\calF, \lfloor \frac{\cd_p(\calF_\alpha)-1}{p-1} \rfloor )$ problem with subset queries to the $\ConDiv{\tfrac{1}{2}}{p}$ problem on normalized monotone functions.
\end{theorem}

\begin{proof}
Fix a prime $p \geq 3$.
For a polynomially computable sequence of set families $\calF = (\calF_\alpha)_{\alpha \in \calA}$, consider an instance of the $\KneserP^p(\calF, \lfloor \frac{\cd_p(\calF_\alpha)-1}{p-1} \rfloor)$ problem with subset queries, that is, an element $\alpha \in \calA$ and Boolean circuits $C:\{0,1\}^{n} \rightarrow [m]$ and $S:\{0,1\}^{n} \times [m] \rightarrow \{0,1\}$ that are supposed to represent a coloring $c:\calF_\alpha \rightarrow [m]$ of the vertices of $K^p(\calF_\alpha)$ where $n = n_\alpha$ and $m = \lfloor \frac{\cd_p(\calF_\alpha)-1}{p-1} \rfloor$ (see Definition~\ref{def:KneserProblemSubset}). Consider the reduction that given such an input returns $m$ Turing machines that compute the valuation functions $v_1, \ldots, v_m : \calB([0,1]) \rightarrow [0,1]$ defined as follows. For each $i \in [m]$ and any $E \in \calB([0,1])$, let $x^E \in [0,1]^n$ denote the vector defined by $x^E_j = n \cdot \mu(E \cap I_j)$ where $I_j = (\frac{j-1}{n},\frac{j}{n})$ for all $j \in [n]$, let $\pi$ be a permutation of $[n]$ such that $x^E_{\pi(1)} \leq x^E_{\pi(2)} \leq \cdots \leq x^E_{\pi(n)}$, and define
\begin{eqnarray*}
v_i(E) = \sum_{j=1}^{n}{(x^E_{\pi(j)} - x^E_{\pi(j-1)}) \cdot S(\{\pi(j), \pi(j+1),\ldots,\pi(n)\},i)},
\end{eqnarray*}
where $x^E_{\pi(0)}=0$.
As shown in the proof of Theorem~\ref{thm:Kneser->CH1}, the functions $v_1,\ldots,v_m$ are normalized, and it is possible to produce in polynomial time efficient Turing machines that compute them. We further define the Lipschitz parameter to be $L=n$. This gives us an instance of the $\ConDiv{\tfrac{1}{2}}{p}$ problem on normalized monotone functions.

For correctness, we show that given a solution for the produced $\ConDiv{\tfrac{1}{2}}{p}$ instance, it is possible to efficiently find a solution for the given instance of $\KneserP^p(\calF, \lfloor \frac{\cd_p(\calF_\alpha)-1}{p-1} \rfloor)$ with subset queries. As shown in the proof of Theorem~\ref{thm:Kneser->CH1}, if the given solution is a violation of one of the valuation functions, then it can be used to efficiently find a violation of the circuit $S$. To avoid repetitions, we omit the details here and focus on the case where the given solution is a partition of the unit interval into $p$ pieces $A_1, \ldots, A_p$ with at most $(p-1) \cdot m$ cuts, such that $|v_i(A_{t}) - v_i(A_{t'})|<\tfrac{1}{2}$ for all $i \in [m]$ and $t,t' \in [p]$.
Given such a solution, it is possible to efficiently construct the set $Y$ of indices $j \in [n]$ for which the sub-interval $I_j$ contains a cut.
Every sub-interval $I_j$ with $j\in [n] \setminus Y$ is fully contained in one of the pieces $A_1, \ldots,A_p$.
For each $t \in [p]$, put $M_t = \{ j \in [n] \mid I_j \subseteq A_t\}$.
By $|Y| \leq (p-1) \cdot m < \cd_p(\calF_\alpha)$, it follows that for some $t_1 \in [p]$, there exists a set $B \in \calF_\alpha$ such that $B \subseteq M_{t_1}$.
Moreover, such a set $B$ can be found efficiently, because the sequence $\calF$ is polynomially computable (see Definition~\ref{def:computable_seq}, Item~\ref{itm:weak2}).
Note that for each $j \in B$, it holds that $\mu(A_{t_1} \cap I_j) = \tfrac{1}{n}$.
We will assume from now on that for any given set $E \in \calB([0,1])$ and color $i \in [m]$, it holds that $v_i(E) = a^E_i$, where $a^E_i$ is the largest value $a \in [0,1]$ such that $S(\{j \in [n] \mid x^E_j \geq a \},i)=1$ if such an $a$ exists, and $a^E_i = 0$ otherwise. As shown in the proof of Theorem~\ref{thm:Kneser->CH1}, if this is not the case, then $E$ and $i$ can be used to efficiently obtain a violation of the circuit $S$.

Let $\ell = c(B)$ denote the color assigned to the set $B$ by the input coloring $c$. This color can be efficiently determined using the circuit $C$. It may be assumed that $S(B,\ell)=1$, as otherwise $B$ and $\ell$ provide a false negative violation of the circuit $S$.
Put $E = \cup_{j \in B}{I_j}$, and observe that $v_\ell(E) = a^E_\ell=1$.
Since $E \subseteq A_{t_1}$, if $v_\ell(A_{t_1})<1$ we obtain a violation of the monotonicity of $v_\ell$.
Otherwise, it holds that $v_\ell(A_{t_1})=1$.
Our assumption on the given solution implies that for each $t \in [p]$, it holds that $|v_\ell(A_{t_1})-v_\ell(A_t)|<\tfrac{1}{2}$ and thus $v_\ell(A_t) = a_\ell^{A_t}> \tfrac{1}{2}$.

Now, for each $t \in [p]$, put $M_t = \{ j \in [n] \mid x_{j}^{A_t} \geq a_{\ell}^{A_t} \}$, and observe that $S(M_t,\ell)=1$.
By Lemma~\ref{lemma:subset_in_F}, it is possible to efficiently find either a violation of $S$ or a set $B_t \in \calF_\alpha$ such that $B_t \subseteq M_t$ and $C(B_t)=\ell$.
If no violation is found, the obtained sets $B_1, \ldots, B_p$ are pairwise disjoint, because for each $j \in [n]$, it holds that $\sum_{t=1}^{p}{\mu(A_{t} \cap I_j)}= \mu(I_j)= \tfrac{1}{n}$, hence at most one $t \in [p]$ can satisfy $\mu(A_{t} \cap I_j)  = \frac{1}{n} \cdot x_{j}^{A_t} \geq \frac{1}{n} \cdot a_{\ell}^{A_t} > \tfrac{1}{2n}$.
This implies that the sets $B_1, \ldots, B_p$ form a monochromatic hyperedge in $K^p(\calF_\alpha)$ and thus a solution for the given instance of $\KneserP^p(\calF, \lfloor \frac{\cd_p(\calF_\alpha)-1}{p-1} \rfloor )$ with subset queries, completing the proof.
\end{proof}

We are ready to derive Theorem~\ref{thm:IntroKneserToCD}.

\begin{proof}[ of Theorem~\ref{thm:IntroKneserToCD}]
Fix a prime $p \geq 3$.
By Lemma~\ref{lemma:weaklyFnk}, the sequence $\calF^{(p)}$ is polynomially computable, allowing us to apply Theorem~\ref{thm:Kneser->CD}, which yields that there exists a polynomial-time reduction from $\KneserP^p(\calF^{(p)}, \lfloor \frac{\cd_p(\calF^{(p)}_{(n,k)})-1}{p-1} \rfloor)$ with subset queries to $\ConDiv{\tfrac{1}{2}}{p}$ on normalized monotone functions.
By Lemma~\ref{lemma:cd_Kneser}, it holds that $\cd_p(\calF^{(p)}_{(n,k)}) = n-p(k-1)$. It thus follows that the $\KneserP^p(\calF^{(p)}, \lfloor \frac{\cd_p(\calF^{(p)}_{(n,k)})-1}{p-1} \rfloor)$ problem coincides with the $\KneserP^p$ problem, and we are done.
\end{proof}

\section{\texorpdfstring{$\KneserP^p$}{Kneser-p} lies in \texorpdfstring{$\PPAp$}{PPA-p}}\label{sec:PPAp}

In this section, we prove the membership of the $\KneserP^p$ problem in the complexity class $\PPAp$ for every prime $p$, confirming Theorem~\ref{thm:IntroPPAp}. The result is proved in two stronger forms through known connections between the chromatic number of Kneser hypergraphs and a $\Z_p$-variant of Tucker's lemma. We then establish limitations on the complexity of variants of the $\KneserP^r$ problem, restricted to colorings with a bounded number of colors. We start by presenting the computational search problem associated with the $\Z_p$-Tucker lemma.

\subsection{The \texorpdfstring{$\ZpTucker$}{Zp-Tucker} Problem}

The definition of the $\ZpTucker$ problem requires a few notations.
For a prime $p$, we denote the elements of the cyclic group $\Z_p$ of order $p$ by $\omega^t$ for $t \in [p]$.
A {\em signed set} over $\Z_p$ is a set whose elements are associated with signs from $\Z_p$. A signed subset of $[n]$ over $\Z_p$ can be represented by a vector $X \in (\Z_p \cup \{0\})^n$, where the subset consists of the elements $j \in [n]$ with $X_j \neq 0$, and the sign of every such $j$ is $X_j$. For two signed sets $X, Y \in (\Z_p \cup \{0\})^n$, we denote by $X \preceq Y$ the fact that for every $j \in [n]$, if $X_j \neq 0$ then $X_j = Y_j$.

\begin{tcolorbox}
\begin{definition}(The $\ZpTucker$ Problem)
For a prime $p$, the $\ZpTucker$ problem is defined as follows.
Its input consists of two integers $n$ and $s$ satisfying $s \leq \lfloor \frac{n-1}{p-1} \rfloor$ along with a Boolean circuit that represents a $\Z_p$-equivariant map $\lambda: (\Z_p \cup \{0\})^n \setminus \{0\}^n \rightarrow \Z_p \times [s]$, that is, a function that maps every nonzero $X \in (\Z_p \cup \{0\})^n$ to a pair $\lambda(X) = (\lambda_1(X), \lambda_2(X))$ in $\Z_p \times [s]$, where for each $t \in [p]$ it holds that $\lambda(\omega^t X)=(\omega^t \lambda_1(X), \lambda_2(X))$.
The goal is to find a chain of $p$ signed sets $X_1 \preceq X_2 \preceq \cdots \preceq X_p$ in $(\Z_p \cup \{0\})^n \setminus \{0\}^n$ that are assigned by $\lambda$ the same absolute value with pairwise distinct signs, that is, for some permutation $\pi$ of $[p]$ and some $\ell \in [s]$, it holds that $\lambda(X_t) = (\omega^{\pi(t)}, \ell)$ for all $t \in [p]$.
\end{definition}
\end{tcolorbox}
\noindent
Note that the assumption that the map $\lambda$ is $\Z_p$-equivariant can be enforced syntactically.
The existence of a solution for every instance of the $\ZpTucker$ problem was proved by Ziegler~\cite{Ziegler02}. Its membership in $\PPAp$, stated below, follows from a much more general result due to Filos-Ratsikas et al.~\cite[Theorem~5.2]{FHSZ21}.
\begin{theorem}[{\cite{FHSZ21}}]\label{thm:TuckerinPPAp}
For every prime $p$, the $\ZpTucker$ problem lies in $\PPAp$.
\end{theorem}

In order to obtain the membership of the $\KneserP^p$ problem in $\PPAp$ for general sequences of set families, we need the following definition.

\begin{definition}\label{def:computable_seq_order}
For a set $\calA$, let $\calF = (\calF_\alpha)_{\alpha \in \calA}$ be a sequence of set families, where for each $\alpha \in \calA$, $\calF_\alpha$ is a family of non-empty subsets of $[n_\alpha]$ for some integer $n_\alpha$.
The sequence $\calF$ is {\em strongly polynomially computable} if it is possible to associate with each family $\calF_\alpha$ a linear order on its members, denoted by $\leq$, such that there exist polynomials $q_1,q_2,q_3$ satisfying that
\begin{enumerate}
  \item\label{itm:seq1s} there exists an algorithm that given an element $\alpha \in \calA$ runs in time $q_1(n_\alpha)$ and returns a Boolean circuit $C_1:\{0,1\}^{2n_\alpha} \rightarrow \{0,1\}$ such that for every pair of sets $B_1,B_2 \in \calF_\alpha$, it holds that $C_1(B_1,B_2)=1$ if and only if $B_1 \leq B_2$,
  \item\label{itm:seq2s} there exists an algorithm that given an element $\alpha \in \calA$ runs in time $q_2(n_\alpha)$ and returns a Boolean circuit $C_2:\{0,1\}^{n_\alpha} \rightarrow \{0,1\}^{n_\alpha}$ such that for every set $D\subseteq [n_\alpha]$, $C_2(D)$ is the minimal subset of $D$, with respect to the order $\leq$, that belongs to the family $\calF_\alpha$ if such a subset exists, and the empty set otherwise, and
  \item\label{itm:seq3s} for every prime $p$, there exists an algorithm that given an element $\alpha \in \calA$ runs in time $q_3(n_\alpha)$ and returns the value of ${\cd}_p(\calF_\alpha)$.
\end{enumerate}
\end{definition}

The following lemma gives simple examples of strongly polynomially computable sequences.

\begin{lemma}\label{lemma:stronglyFnk}
For every integer $r \geq 2$, the sequence $\calF^{(r)}$ is strongly polynomially computable.
\end{lemma}

\begin{proof}
Fix an integer $r \geq 2$, and recall that $\calF^{(r)}_{(n,k)} = \binom{[n]}{k}$ for all $(n,k) \in \calA^{(r)}$.
We associate with the sets of $\calF^{(r)}_{(n,k)}$ the linear order $\leq$, defined by $B_1 \leq B_2$ if $B_1 = B_2$ or the smallest element of $B_1 \triangle B_2$ belongs to $B_1$.
To show that $\calF^{(r)}$ is strongly polynomially computable, we verify the three conditions of Definition~\ref{def:computable_seq_order}.
For Item~\ref{itm:seq1s}, consider the algorithm that, given a pair $(n,k) \in \calA^{(r)}$, produces a Boolean circuit that on input of two sets $B_1,B_2 \in \binom{[n]}{k}$, checks whether $B_1 \leq B_2$, that is, either their characteristic vectors are equal, or the first entry in which they differ is $1$ in the vector of $B_1$.
For Item~\ref{itm:seq2s}, consider the algorithm, that given a pair $(n,k) \in \calA^{(r)}$, produces a Boolean circuit that on an input set $D \subseteq [n]$, returns the empty set if $|D|<k$, and the set of the $k$ smallest elements of $D$ otherwise.
Clearly, both algorithms can be implemented in time polynomial in $n$.
For Item~\ref{itm:seq3s}, observe that the desired algorithm follows from Lemma~\ref{lemma:cd_Kneser}. This completes the proof.
\end{proof}

\subsection{From \texorpdfstring{$\KneserP^p$}{Kneser-p} to \texorpdfstring{$\ZpTucker$}{Zp-Tucker}}

The following theorem asserts that the $\KneserP^p(\calF,m)$ problem is efficiently reducible to the $\ZpTucker$ problem for every strongly polynomially computable sequence $\calF$, whenever the number of colors $m$ is smaller than the bound given by Theorem~\ref{thm:IntroKriz}. Its proof verifies that a mathematical argument of Ziegler~\cite{Ziegler02} can be transformed into an efficient reduction. We present it here with the details for completeness.

\begin{theorem}\label{thm:Kneser^p(F)->Tucker}
Let $p$ be a prime, and let $\calF = (\calF_\alpha)_{\alpha \in \calA}$ be a strongly polynomially computable sequence of set families.
Then, $\KneserP^p(\calF, \lfloor \frac{\cd_p(\calF_\alpha)-1}{p-1} \rfloor )$ is polynomial-time reducible to $\ZpTucker$.
\end{theorem}

\begin{proof}
Fix a prime $p$ and a strongly polynomially computable sequence of set families $\calF = (\calF_\alpha)_{\alpha \in \calA}$.
Consider an instance of $\KneserP^p(\calF, \lfloor \frac{\cd_p(\calF_\alpha)-1}{p-1} \rfloor )$, that is, an element $\alpha \in \calA$ along with a Boolean circuit that represents a coloring $c: \calF_\alpha \rightarrow [m]$ of the vertices of $K^p(\calF_\alpha)$, where $\calF_\alpha$ is a family of subsets of $[n]$ for some $n=n_\alpha$ and $m = m(\alpha) = \lfloor \frac{\cd_p(\calF_\alpha)-1}{p-1} \rfloor$.
It may be assumed that $p-1$ divides $n-1$. Indeed, otherwise we increase the size $n$ of the ground set of $\calF_\alpha$ by at most $p-2$, so that this condition will be satisfied. This change of $n$ does not affect the chromatic number of $K^p(\calF_\alpha)$ nor the value of ${\cd}_p(\calF_\alpha)$.

Consider the function $\lambda: (\Z_p \cup \{0\})^n \setminus \{0\}^n \rightarrow \Z_p \times [s]$ defined as follows. For a given $X \in (\Z_p \cup \{0\})^n \setminus \{0\}^n$, put $X^t = \{j \in [n] \mid X_j = \omega^t \}$ for each $t \in [p]$, and let $|X| = \sum_{t=1}^{p}{|X^t|}$ denote the number of nonzero elements in $X$. We define $\lambda(X) = (\lambda_1(X), \lambda_2(X))$ according to the following two cases.
\begin{enumerate}
  \item If $|X| \leq n-\cd_p(\calF_\alpha)$, then let $t \in [p]$ denote the index of the set $X^t$ that includes the smallest element of $\cup_{t' \in [p]}{X^{t'}}$, and define
  $\lambda(X) = (\omega^t, \lceil \frac{|X|}{p-1} \rceil)$.
  \item If $|X| > n-\cd_p(\calF_\alpha)$, then it follows that some $X^t$ contains a set of $\calF_\alpha$. Among the sets of $\calF_\alpha$ that are contained in one of the sets $X^1, \ldots, X^p$, let $B$ denote the minimal one with respect to the linear order associated with $\calF_\alpha$. Let $t \in [p]$ denote the index satisfying $B\subseteq X^t$, and define $\lambda(X) = (\omega^t, c(B)+\lceil \frac{n-\cd_p(\calF_\alpha)}{p-1} \rceil )$.
\end{enumerate}
Observe that the function $\lambda$ returns pairs in $\Z_p \times [s]$ for
\[s = m + \Big \lceil \frac{n-\cd_p(\calF_\alpha)}{p-1} \Big \rceil = \Big \lfloor \frac{\cd_p(\calF_\alpha)-1}{p-1} \Big \rfloor + \Big \lfloor \frac{n-\cd_p(\calF_\alpha)-1}{p-1} \Big \rfloor +1 = \frac{n-1}{p-1},\]
where in the last equality we use the fact that $p-1$ divides $n-1$.
We claim that $\lambda$ is $\Z_p$-equivariant. To see this, consider some $X \in (\Z_p \cup \{0\})^n \setminus \{0\}^n$ and $t \in [p]$, and notice that $|X| = |\omega^t X|$. Notice further that for each $t' \in [p]$, the indices of the entries in $X$ that are equal to $\omega^{t'}$ are precisely the indices of the entries in $\omega^t X$ that are equal to $\omega^{t+t'}$. In particular, if the first nonzero element in $X$ is $\omega^{t'}$ then the first nonzero element in $\omega^t X$ is $\omega^{t+t'}$. By considering the two cases in the definition of $\lambda$, it follows that $\lambda(\omega^t X) = (\omega^t \lambda_1(X),\lambda_2(X))$, as required.
The assumption that $\calF$ is strongly polynomially computable (see Definition~\ref{def:computable_seq_order}) implies that given a Boolean circuit that represents the coloring $c$, it is possible to construct in polynomial time a Boolean circuit that represents the function $\lambda$.

To prove the correctness of the reduction, consider a solution for the produced instance of $\ZpTucker$, that is, a chain of $p$ signed sets $X_1 \preceq X_2 \preceq \cdots \preceq X_p \in (\Z_p \cup \{0\})^n \setminus \{0\}^n$ for which there exist a permutation $\pi$ of $[p]$ and an $\ell \in [s]$, such that $\lambda(X_t) = (\omega^{\pi(t)}, \ell)$ for all $t \in [p]$. We first claim that $\ell > \lceil \frac{n-\cd_p(\calF_\alpha)}{p-1} \rceil$. To see this, suppose for contradiction that $\ell \leq \lceil \frac{n-\cd_p(\calF_\alpha)}{p-1} \rceil$, and observe that the definition of $\lambda$ implies that $\ell = \lceil \frac{|X_t|}{p-1} \rceil$ for all $t \in [p]$. It follows that for some indices $t_1 < t_2 \in [p]$, it holds that $|X_{t_1}| = |X_{t_2}|$. However, by $X_{t_1} \preceq X_{t_2}$, it follows that $X_{t_1} = X_{t_2}$, in contradiction to the fact that $\lambda_1(X_{t_1}) \neq \lambda_1(X_{t_2})$.

Finally, using $\ell > \lceil \frac{n-\cd_p(\calF_\alpha)}{p-1} \rceil$, the definition of $\lambda$ implies that for every $t \in [p]$, there exists a subset of $X_t^{\pi(t)}$ that belongs to $\calF_\alpha$. For each $t \in [p]$, let $B_t$ denote the minimal set of $\calF_\alpha$ that is contained in $X_t^{\pi(t)}$, with respect to the order associated with $\calF_\alpha$. By the definition of $\lambda$, we have $\ell = c(B_t) + \lceil \frac{n-\cd_p(\calF_\alpha)}{p-1} \rceil$ for all $t \in [p]$. The assumption that $\calF$ is strongly polynomially computable allows us to find the sets $B_1, \ldots, B_p$ efficiently (see Definition~\ref{def:computable_seq_order}, Item~\ref{itm:seq2s}).
For any $t_1 < t_2 \in [p]$, the fact that $X_{t_1} \preceq X_{t_2}$ implies that $B_{t_1}$ and $B_{t_2}$ are disjoint. It thus follows that the sets $B_1, \ldots, B_p$ form a monochromatic hyperedge of the coloring $c$ of $K^p(\calF_\alpha)$.
Since this monochromatic hyperedge can be found in polynomial time given the signed sets $X_1,X_2, \ldots, X_p$, we are done.
\end{proof}

By combining Theorem~\ref{thm:TuckerinPPAp} with Theorem~\ref{thm:Kneser^p(F)->Tucker}, we derive the following corollary.
\begin{corollary}\label{cor:Kneser^p_cd_inPPAp}
Let $p$ be a prime, and let $\calF = (\calF_\alpha)_{\alpha \in \calA}$ be a strongly polynomially computable sequence of set families.
Then, the $\KneserP^p(\calF, \lfloor \frac{\cd_p(\calF_\alpha)-1}{p-1} \rfloor )$ problem lies in $\PPAp$.
\end{corollary}

As a special case of Corollary~\ref{cor:Kneser^p_cd_inPPAp}, we derive Theorem~\ref{thm:IntroPPAp}.

\begin{proof}[ of Theorem~\ref{thm:IntroPPAp}]
Fix a prime $p$.
By Lemma~\ref{lemma:stronglyFnk}, the sequence $\calF^{(p)}$ is strongly polynomially computable, hence by Corollary~\ref{cor:Kneser^p_cd_inPPAp}, the $\KneserP^p(\calF^{(p)}, \lfloor \frac{\cd_p(\calF^{(p)}_{(n,k)})-1}{p-1} \rfloor)$ problem lies in $\PPAp$.
By Lemma~\ref{lemma:cd_Kneser}, it holds that $\cd_p(\calF^{(p)}_{(n,k)}) = n-p(k-1)$, hence the latter problem coincides with the $\KneserP^p$ problem, and we are done.
\end{proof}

We next consider the $\KneserPAstab^p$ problem associated with the hypergraph $K^p(n,k)_{\AlmostTwoStab}$, that is, the sub-hypergraph of $K^p(n,k)$ induced by the almost stable $k$-subsets of $[n]$. Corollary~\ref{cor:Kneser^p_cd_inPPAp} can be applied to obtain a membership result in $\PPAp$ for this setting, however, it does not give the largest possible number of colors. To obtain the result with an optimal number of colors, which is smaller than the chromatic number only by one, we apply a modified argument of Meunier~\cite{Meunier11}, verifying that it provides an efficient reduction. Let us mention, though, that the proof of~\cite{Meunier11} applies a slightly different variant of the $\Z_p$-Tucker lemma, whose proof relies on Dold's theorem~\cite{Dold83}. The proof below uses the version of the lemma that corresponds to our definition of the $\ZpTucker$ problem, which lies in $\PPAp$. See~\cite[Remark~1]{FHSZ21} for a discussion on the computational aspects of Dold's theorem.

\begin{theorem}\label{thm:Schrijver^p->Tucker}
For every prime $p$, $\KneserPAstab^p$ is polynomial-time reducible to $\ZpTucker$.
\end{theorem}

\begin{proof}
Fix a prime $p$, and consider an instance of $\KneserPAstab^p$, that is, integers $n$ and $k$ satisfying $n \geq p \cdot k$ along with a Boolean circuit that represents a coloring $c: \binom{[n]}{k}_{\AlmostTwoStab} \rightarrow [m]$ of the vertices of $K^p(n,k)_{\AlmostTwoStab}$ where $m = \lfloor \frac{n-p(k-1)-1}{p-1} \rfloor$. Let $a$ denote the smallest non-negative integer such that $p-1$ divides $p(k-1)+a$, and set $n'=n+a$.

Consider the function $\lambda: (\Z_p \cup \{0\})^{n'} \setminus \{0\}^{n'} \rightarrow \Z_p \times [s]$ defined as follows. For a given $X$ in $(\Z_p \cup \{0\})^{n'} \setminus \{0\}^{n'}$, let $\alt(X)$ denote the length of a longest alternating subsequence of nonzero terms of $X$, namely, the largest integer $r$ for which there exist indices $i_1 < \cdots < i_r$ such that $X_{i_j} \neq 0$ for all $j \in [r]$ and $X_{i_j} \neq X_{i_{j+1}}$ for all $j \in [r-1]$.
Put $X^t = \{j \in [n'] \mid X_j = \omega^t \}$ for each $t\in [p]$, and define $\lambda(X) = (\lambda_1(X), \lambda_2(X))$ according to the following two cases.
\begin{enumerate}
  \item If $\alt(X) \leq p(k-1)+a$, then let $t \in [p]$ denote the index of the set $X^t$ that includes the smallest element of $\cup_{t' \in [p]}{X^{t'}}$, and define $\lambda(X) = (\omega^t, \lceil \frac{\alt(X)}{p-1} \rceil)$.
  \item If $\alt(X) > p(k-1)+a$, then consider a longest alternating subsequence of nonzero terms of $X$. Since its length is at least $p(k-1)+a+1$, it must contain an alternating subsequence of $p(k-1)+1$ nonzero terms of $X$ that lie in the first $n$ entries of $X$. This implies that some set $X^t$ contains an almost stable $k$-subset of $[n]$. Among the almost stable $k$-subsets of $[n]$ that are contained in one of the sets $X^1, \ldots, X^p$, let $B$ denote the minimal one with respect to some linear order associated with $\binom{[n]}{k}_{\AlmostTwoStab}$ (say, $B_1 \leq B_2$ if $B_1 = B_2$ or the smallest element of $B_1 \triangle B_2$ belongs to $B_1$). Let $t \in [p]$ denote the index that satisfies $B \subseteq X^t$, and define $\lambda(X) = (\omega^t, c(B)+\frac{p(k-1)+a}{p-1})$.
\end{enumerate}
Observe that the function $\lambda$ returns pairs in $\Z_p \times [s]$ for
\[s = m + \frac{p(k-1)+a}{p-1} = \Big \lfloor \frac{n-p(k-1)-1}{p-1} \Big \rfloor + \frac{p(k-1)+a}{p-1} = \Big \lfloor \frac{n+a-1}{p-1} \Big \rfloor = \Big \lfloor \frac{n'-1}{p-1} \Big \rfloor.\]
We claim that $\lambda$ is $\Z_p$-equivariant. This indeed follows by considering the two cases in the definition of $\lambda$, using the observation that for all $X \in (\Z_p \cup \{0\})^{n'} \setminus \{0\}^{n'}$ and $t \in [p]$, it holds that $\alt(X) = \alt(\omega^t X)$.
It can be verified that given a Boolean circuit that represents the coloring $c$, it is possible to construct in polynomial time a Boolean circuit that represents the function $\lambda$.

To prove the correctness of the reduction, consider a solution for the produced instance of $\ZpTucker$, that is, a chain of $p$ signed sets $X_1 \preceq X_2 \preceq \cdots \preceq X_p \in (\Z_p \cup \{0\})^{n'} \setminus \{0\}^{n'}$ for which there exist a permutation $\pi$ of $[p]$ and an $\ell \in [s]$, such that $\lambda(X_t) = (\omega^{\pi(t)}, \ell)$ for all $t \in [p]$. We first claim that $\ell > \frac{p(k-1)+a}{p-1}$. To see this, suppose for contradiction that $\ell \leq \frac{p(k-1)+a}{p-1}$, and observe that the definition of $\lambda$ implies that $\ell = \lceil \frac{\alt(X_t)}{p-1} \rceil$ for all $t \in [p]$. It follows that some indices $t_1 < t_2 \in [p]$ satisfy that $\alt(X_{t_1}) = \alt(X_{t_2})$. However, by $X_{t_1} \preceq X_{t_2}$, a longest alternating subsequence of nonzero terms of $X_{t_1}$ is also a subsequence of $X_{t_2}$. This implies that the first nonzero elements in $X_{t_1}$ and $X_{t_2}$ are equal, in contradiction to the fact that $\lambda_1(X_{t_1}) \neq \lambda_1(X_{t_2})$.

Finally, using $\ell > \frac{p(k-1)+a}{p-1}$, the definition of $\lambda$ implies that for every $t \in [p]$, there exists a subset of $X_t^{\pi(t)}$ that belongs to $\binom{[n]}{k}_{\AlmostTwoStab}$. For each $t \in [p]$, let $B_t$ denote the minimal subset of $X_t^{\pi(t)}$ that belongs to $\binom{[n]}{k}_{\AlmostTwoStab}$. By the definition of $\lambda$, we have $\ell = c(B_t) + \frac{p(k-1)+a}{p-1}$ for all $t \in [p]$. For any $t_1 < t_2 \in [p]$, the fact that $X_{t_1} \preceq X_{t_2}$ implies that $B_{t_1}$ and $B_{t_2}$ are disjoint. It thus follows that the sets $B_1, \ldots, B_p$ form a monochromatic hyperedge of the coloring $c$ of $K^p(n,k)_{\AlmostTwoStab}$.
It can be verified that this monochromatic hyperedge can be found in polynomial time given the signed sets $X_1,X_2, \ldots, X_p$, so we are done.
\end{proof}

By combining Theorem~\ref{thm:TuckerinPPAp} with Theorem~\ref{thm:Schrijver^p->Tucker}, the proof of Theorem~\ref{thm:IntroPPApSch} is completed.

\subsection{Limitations on the Complexity of \texorpdfstring{$\KneserP^r$}{Kneser-r} Problems}

We next show that the results from the previous section imply limitations on the complexity of variants of the $\KneserP^r$ problem, restricted to colorings with a bounded number of colors.
We start with the following simple lemma, which says that the $\KneserP^{r}(\calF,m)$ problem does not become easier when $r$ increases.
\begin{lemma}\label{lemma:Kneser_p1_p2}
Let $r_1 \leq r_2$ be integers, let $\calF = (\calF_\alpha)_{\alpha \in \calA}$ be a sequence of set families, and let $m: \calA \rightarrow \N$ be a function such that $m(\alpha) < \chi(K^{r_2}(\calF_\alpha))$ for all $\alpha \in \calA$. Then, $\KneserP^{r_1}(\calF,m)$ is polynomial-time reducible to $\KneserP^{r_2}(\calF,m)$.
\end{lemma}

\begin{proof}
Fix two integers $r_1 \leq r_2$.
Consider an instance of the $\KneserP^{r_1}(\calF,m)$ problem, that is, an element $\alpha \in \calA$ and a Boolean circuit that represents a coloring of $\calF_\alpha$ with $m(\alpha)$ colors.
We simply apply the identity reduction to the $\KneserP^{r_2}(\calF,m)$ problem.
A solution for the obtained instance of $\KneserP^{r_2}(\calF,m)$, whose existence follows by $m(\alpha) < \chi(K^{r_2}(\calF_\alpha))$, is a collection of $r_2$ pairwise disjoint sets of $\calF_\alpha$ with the same color.
By $r_1 \leq r_2$, any $r_1$ sets from this collection form a solution for the same input as an instance of $\KneserP^{r_1}(\calF,m)$. The correctness of the reduction follows.
\end{proof}

By combining Lemma~\ref{lemma:Kneser_p1_p2} with Theorem~\ref{thm:IntroPPAp}, we obtain the following result.
\begin{theorem}\label{thm:limit_Kneser}
For every integer $r$ and for every prime $p$ such that $r \leq p$, $\KneserP^{r}(\calF^{(p)}, \lfloor \frac{n-p(k-1)-1}{p-1} \rfloor )$ lies in $\PPAp$.
\end{theorem}
\begin{proof}
Fix an integer $r$ and a prime $p$ such that $r \leq p$.
Put $m(n,k) = \lfloor \frac{n-p(k-1)-1}{p-1} \rfloor $, and apply Lemma~\ref{lemma:Kneser_p1_p2} to obtain that $\KneserP^{r}(\calF^{(p)},m)$ is polynomial-time reducible to $\KneserP^{p}(\calF^{(p)},m)$. The latter coincides with the $\KneserP^{p}$ problem, which, by Theorem~\ref{thm:IntroPPAp}, lies in $\PPAp$.
It thus follows that $\KneserP^{r}(\calF^{(p)},m)$ lies in $\PPAp$, as required.
\end{proof}

Theorem~\ref{thm:limit_Kneser} yields, for any integer $r$, a limitation on the complexity of the $\KneserP^{r}$ problem, restricted to colorings with a bounded number of colors. For example, consider the $\KneserP$ problem, which asks to find a monochromatic edge in a graph $K(n,k)$ colored with $n-2k+1$ colors, and recall that it lies in $\PPA$. By Theorem~\ref{thm:limit_Kneser}, applied with $r=2$ and $p=3$, the $\KneserP(n,k,\lfloor \frac{n-3k+2}{2} \rfloor)$ problem, which asks to find a monochromatic edge in a graph $K(n,k)$ colored with only $\lfloor \frac{n-3k+2}{2} \rfloor$ colors, lies in $\PPAthree$. This implies that the latter problem is not $\PPA$-hard, unless $\PPA \subseteq \PPAthree$. We next present analogue consequences for the $\SchrijverP$ problem.

We need the following simple lemma.
\begin{lemma}\label{lemma:Stab->AStab}
For every integer $r$, $\KneserPstab^r(n,k,\lfloor \frac{n-rk}{r-1} \rfloor)$ is polynomial-time reducible to $\KneserPAstab^r$.
\end{lemma}

\begin{proof}
For an integer $r$, consider an instance of $\KneserPstab^r(n,k,\lfloor \frac{n-rk}{r-1} \rfloor)$, that is, integers $n$ and $k$ with $n \geq r \cdot k$ along with a coloring $c: \binom{[n]}{k}_\TwoStab \rightarrow [m]$ of the vertices of $K^r(n,k)_\TwoStab$ where $m = \lfloor \frac{n-rk}{r-1} \rfloor$. We define a coloring $c': {\binom{[n]}{k}}_{\AlmostTwoStab} \rightarrow [m+1]$ of the vertices of ${K^r(n,k)}_{\AlmostTwoStab}$ as follows. For every $A \in \binom{[n]}{k}_{\AlmostTwoStab}$, if $n \in A$ then $c'(A) = m+1$. Otherwise, it holds that $A \in \binom{[n]}{k}_\TwoStab$, and we define $c'(A) = c(A) \in [m]$. Note that the number of colors used by $c'$ is $m+1 = \lfloor \frac{n-r(k-1)-1}{r-1} \rfloor$, as needed for an instance of the $\KneserPAstab^r$ problem. Notice that given a Boolean circuit that represents the coloring $c$, it is possible to construct in polynomial time a Boolean circuit that represents the coloring $c'$.

For correctness, consider a solution for the produced $\KneserPAstab^r$ instance, that is, a monochromatic hyperedge of the coloring $c'$ of ${K^r(n,k)}_{\AlmostTwoStab}$. By the definition of $c'$, the vertices colored $m+1$ are pairwise intersecting, hence the vertices of the monochromatic hyperedge are not colored $m+1$, and thus do not include the element $n$. Since an almost stable subset of $[n]$ that does not include $n$ is stable, this is a monochromatic hyperedge of the coloring $c$ of $K^r(n,k)_\TwoStab$ and thus a solution for the given instance of $\KneserPstab^r(n,k,\lfloor \frac{n-rk}{r-1} \rfloor)$, so we are done.
\end{proof}

By combining Lemma~\ref{lemma:Stab->AStab} with Theorem~\ref{thm:IntroPPApSch}, we derive the following.
\begin{corollary}\label{cor:Stab-PPAp}
For every prime $p$, the $\KneserPstab^p(n,k,\lfloor \frac{n-pk}{p-1} \rfloor)$ problem lies in $\PPAp$.
\end{corollary}

We further derive the following corollary.
\begin{corollary}\label{cor:SchriverPPA3}
The $\SchrijverP(n,k,\lfloor \frac{n-3k}{2} \rfloor)$ problem lies in $\PPAthree$.
\end{corollary}

\begin{proof}
Put $m(n,k) = \lfloor \frac{n-3k}{2} \rfloor$.
By Lemma~\ref{lemma:Kneser_p1_p2}, the $\SchrijverP(n,k,m)$ problem, which can be written as $\KneserPstab(n,k,m)$, is polynomial-time reducible to $\KneserPstab^{3}(n,k,m)$.
By Corollary~\ref{cor:Stab-PPAp}, the latter lies in $\PPAthree$. It thus follows that $\SchrijverP(n,k,m)$ lies in $\PPAthree$ as well.
\end{proof}

We finally state the following consequence of Corollary~\ref{cor:SchriverPPA3} regarding the $\SchrijverP(n,k,m)$ problem with the function $m(n,k) = \lfloor n/2 \rfloor-2k+1$ considered in~\cite{Haviv23}.
We use here the fact that for all integers $n$ and $k \geq 2$, it holds that $\lfloor n/2 \rfloor-2k+1 \leq \lfloor \frac{n-3k}{2} \rfloor$.

\begin{corollary}
The $\SchrijverP(n,k,\lfloor n/2 \rfloor-2k+1)$ problem lies in $\PPAthree$.
\end{corollary}

\section*{Acknowledgments}
We are grateful to Aris Filos-Ratsikas for a fruitful discussion and for clarifications on~\cite{DeligkasFH22} and to the anonymous referees for their useful and constructive feedback.

\bibliographystyle{abbrv}
\bibliography{kneser-hyper}

\end{document}